\documentclass[11pt,english]{article}
\usepackage[T1]{fontenc}

\usepackage{graphicx, subfigure}
\expandafter\def\csname ver@subfig.sty\endcsname{}
\usepackage{subcaption}

\usepackage[latin9]{inputenc}
\usepackage{geometry}
\usepackage{amsmath}
\usepackage{appendix}
\usepackage{amssymb}
\usepackage{amsfonts}
\usepackage{esint}
\usepackage{makecell,multirow,diagbox}
\usepackage{mathtools}
\usepackage{epstopdf}
\usepackage[latin9]{inputenc}
\usepackage{pict2e}
\usepackage{stmaryrd}
\usepackage{esint}
\usepackage[all,cmtip]{xy}
\usepackage{mathtools}
\usepackage{float}
\usepackage{setspace}
\usepackage{authblk}
\usepackage{amsthm}
\usepackage{mathrsfs}
\usepackage{stmaryrd}
\usepackage{color}
\usepackage{tikz}
\usepackage{verbatim}
\usetikzlibrary{arrows,backgrounds,snakes,shapes}
\usepackage[all,cmtip]{xy}
\usepackage{algorithm}
\usepackage{algorithmic}

\theoremstyle{plain}
\theoremstyle{plain}
\newtheorem{pro}{Proposition}
\newtheorem{remark}{Remark}
\usepackage{color}
\definecolor{marin}{rgb} {0., 0.3, 0.7}
\definecolor{rouge}{rgb} {0.8, 0., 0.}
\definecolor{sepia}{rgb} {0.8, 0.5, 0.}
\usepackage[colorlinks,citecolor=sepia,linkcolor=marin,bookmarksopen,bookmarksnumbered]{hyperref}

\theoremstyle{definition}

\DeclareSymbolFont{largesymbol}{OMX}{yhex}{m}{n}
\DeclareMathAccent{\Widehat}{\mathord}{largesymbol}{"62}

\newcommand{\derv}[2]{ \frac{\partial #1}{\partial #2} }

\makeatletter

\usepackage{babel}
\usepackage{subfig}
\graphicspath{{Figure/}}

\makeatother
\DeclareMathSizes{14}{14}{9.8}{7}

\usepackage{lipsum}

\begin{document}

\title{\textbf{Ultrafast dynamics of a spin-polarized electron plasma with magnetic ions}}
\date{}
\author[1]{\small \textbf{Benjamin Bakri}} \author[2]{\small \textbf{Nicolas Crouseilles}} \author[1]{\small \textbf{Paul-Antoine Hervieux}}
\author[3]{\small \textbf{Xue Hong}} 
\author[1]{\small \textbf{Giovanni Manfredi}}
\affil[1] {Universit\'e de Strasbourg, CNRS, Institut de Physique et Chimie des Mat\'eriaux de Strasbourg, UMR 7504, F-67000 Strasbourg, France}
\affil[2] {Universit\'e de Rennes, Inria Rennes (Mingus team) and IRMAR UMR CNRS 6625, F-35042 Rennes, France \& ENS Rennes}
\affil[3]{Department of Radiation Oncology, University of Kansas Medical Center, USA}

\maketitle
\begin{abstract}
We construct a mean-field model that describes the nonlinear dynamics of a spin-polarized electron gas interacting with fixed, positively-charged ions possessing a magnetic moment that evolves in time. The mobile electrons are modeled by a four-component distribution function in the two-dimensional phase space $(x,v)$, obeying a Vlasov-Poisson set of equations. The ions are modeled by a Landau-Lifshitz equation for their spin density, which contains ion-ion and electron-ion magnetic exchange terms. 
We perform a linear response study of the coupled Vlasov-Poisson-Landau-Lifshitz (VPLL) equations for the case of a Maxwell-Boltzmann equilibrium, focussing in particular on the spin dispersion relation. Condition of stability or instability for the spin modes are identified, which depend essentially on the electron spin polarization rate $\eta$ and the electron-ion magnetic coupling constant $K$.
We also develop an Eulerian grid-based computational code for the fully nonlinear VPLL equations, based on the geometric Hamiltonian method first developed in \cite{Crouseilles2023}. This technique allows us to achieve great accuracy for the conserved quantities, such as the modulus of the ion spin vector and the total energy.
Numerical tests in the linear regime are in accordance with the estimations of the linear response theory. 
For two-stream equilibria, we study the interplay of instabilities occurring in both the charge and the spin sectors.
The set of parameters used in the simulations,  with densities close to those of solids ($\approx 10^{29} \rm m^{-3}$) and temperatures of the order of 10~eV, may be relevant to the warm dense matter regime  appearing in some inertial fusion experiments.
\end{abstract}
\setcounter{tocdepth}{1} 

\tableofcontents

\section{Introduction} 

The interaction of coherent electromagnetic radiation (laser light) with matter is a well-established field within various branches of physics, particularly condensed-matter and nanophysics, where laser pulses are often employed to study how electrons behave on extremely short time scales (femto- or atto-seconds). Indeed, the most common electronic resonance found in metals -- the plasmon resonance -- occurs within the femtosecond time scale. This makes ultrafast laser pulses an essential tool for experimental investigations into the collective behavior of electrons in metals. 

In plasma physics, laser-plasma interactions are essential for the development of inertial fusion (triggered by powerful laser pulses) and laser-plasma accelerators (which rely on the acceleration of charged particles by plasma waves). They are also crucial in the study of warm dense matter (WDM), a state of matter that is at the frontier between solids and dense plasmas, where ultrafast nonequilibrium dynamics have been recently accessed thanks to subpicosecond laser pulses \cite{Falk2018}.

However, in addition to their electric charge, electrons also possess an intrinsic magnetic moment, i.e., a spin. Utilizing the electron spin as a vector to code and transfer information is at the core of the emerging field of spintronics.
In nanophysics, spin effects are at the core of the ultrafast demagnetization observed in ferromagnetic thin films irradiated with femtosecond laser pulses \cite{beaurepaire1996ultrafast,bigot2009coherent,bigot2013ultrafast}.
Despite intense investigations, such ultrafast demagnetization is not yet fully understood, although the spin-orbit interaction \cite{Krieger2015, spin_orbit2, Hinsch2012}, spin currents \cite{spin_current1, spin_current2,Hurst2014} and superdiffusive electron transport \cite{Battiato2010} appear to play a significant role. 

The exploration of spin-dependent effects in plasma physics is a relatively new area of study. Nonetheless, it is now possible to generate and precisely control polarized electron beams with high spin polarization in laboratory settings \cite{Wu2019,Wu2020,Nie2021}. Theoretical studies on polarized plasmas have been revitalized in recent years \cite{Hurst2014,Hurst2017,Zamanian2010njp,Zamanian2010,Morandi2014}, although some early developments date back to the 1980s \cite{Cowley1986PoP}. Notably, Brodin et al. \cite{Brodin2013JPP} have formulated a particle-in-cell (PIC) code that incorporates the magnetic dipole force and magnetization currents related to the electron spin. PIC methods for particles with spin have also been developed for applications in the field of laser-plasma interactions \cite{FeiLi2021JCP}.

Within the condensed matter and nanophysics communities, most research on ultrafast spin dynamics has relied on wavefunction based methods, particularly time-dependent density functional theory, augmented in order to incorporate spin effects (spin-TDDFT) \cite{Krieger2015,Sinha2020,Yin2009,MHC2010}. Spin-TDDFT models have also been utilized to study spin effects in dense plasmas in the WDM regime \cite{Bonitz2020}.

In a recent series of papers \cite{Crouseilles2021,Crouseilles2023,Manfredi2023}, we have proposed an alternative approach based on Wigner functions, which represent electronic quantum states through a pseudo-probability distribution in the classical phase space. The corresponding Wigner evolution equation reduces to the standard Vlasov equation of classical plasma physics. For spin-1/2 particles, such as electrons, one can construct a semi-classical model, where the orbital motion (i.e., the trajectories in the phase space) is treated classically while the spin is kept as a quantum-mechanical variable. For a review of methods based on Wigner functions, see \cite{Manfredi2019}. 

Among these phase space models, two families can be distinguished: on the one side, Vlasov models that use a scalar distribution function on an extended phase-space $(x,v,s)$ 
where $x$ and $v$ are the position and velocity of the electron, while $s$ denotes the spin variable \cite{Zamanian2010njp, marklund2010spin, scalar3, brodin2008effects, marklund2010spin, brodin2008effects}; on the other side, models using a multi-component distribution function $f_\ell, \, (\ell=0,3)$ with values in the standard phase space $(x, v)$. 
These two approaches are almost, although not exactly, equivalent (see our detailed discussion in \cite{Crouseilles2023} for further clarifications). Hereafter, we will name these approaches respectively as "scalar" and "vectorial". Note that, for both of them, the orbital motion is classical while the spin is a fully quantum variable. 
The numerical approximation of these models requires different techniques. Indeed, the scalar version involves an extended phase space of dimension 8, which naturally leads to consider PIC techniques as the method of choice \cite{Crouseilles2021, LI2023111733}; in contrast, the vectorial approach is more easily amenable to grid-based methods \cite{Crouseilles2023}.  

In previous works \cite{Crouseilles2021,Crouseilles2023,Manfredi2023}, we had only considered the dynamics of the mobile (itinerant) electrons, whereas the ions only acted as an immobile neutralizing background. However, in ferromagnets most of the magnetic properties are due to the fixed ions, which account for approximately $95\%$ of the magnetization of the material, whereas only the remaining $\approx 5\%$ can be attributed to the mobile  electrons. 
In the present work, the ions are still fixed (because their orbital response occurs on much longer timescales), but their spin is allowed to evolve in time according to the Landau-Lifshitz (LL) equation. The latter describes the precession motion of a magnetic moment in an effective magnetic field, which can be either an external one or the field created by the spin of the itinerant electrons. In turns, the ions  generate a magnetic field which acts on the spin of the electrons.  The ions also interact among each other through a Heisenberg-type magnetic-exchange interaction, while the electrons feel the usual self-consistent electric field.

Overall, the nonlinear Vlasov-Poisson-Landau-Lifshitz (VPLL) equations describe the coupling between the itinerant magnetism generated by the mobile electrons, represented by a vector distribution function $(f_0, {\bf f})(t, x, v)\in\mathbb{R}^4$, and the fixed magnetism carried by the motionless ions, represented by their local spin ${\bf S}(t, x) \in \mathbb{S}^2$. It can be viewed as a spin-extended version of the usual Vlasov-Poisson model with fixed ions.
An earlier version of this model -- employing a more rudimentary numerical technique -- was used in \cite{Hurst2018} to study spin current generation in thin nickel films.  
Here, we will mainly consider a parameter range relevant to WDM \cite{Bonitz2020}, with densities close to those of solids ($\approx 10^{29} \rm m^{-3}$) and temperatures of the order of 10~eV. For these conditions, the electron plasma is weakly degenerate ($T_e \approx T_F$, where $T_F$ is the Fermi temperature), so that its equilibrium can be characterized with relatively good accuracy by a Maxwell-Boltzmann distribution. The ions are fixed and non-degenerate. 

The model is described mathematically by a set of coupled nonlinear partial differential equations (PDEs). The design of efficient scheme for a system of PDEs is not easy and one possible strategy is to make use of a splitting algorithm.
When the system under consideration enjoys a Hamiltonian structure, a systematic way to proceed relies on the Hamiltonian splitting \cite{splitting, casas, laser2020, CRESTETTO2022110857}. It turns out that the VPLL equations enjoy a Poisson structure which motivates the use of Hamiltonian time splitting. Following previous 
development of geometric numerical method for Vlasov-type equations \cite{splitting, laser2020, CRESTETTO2022110857}, the Hamiltonian splitting applied to the VPLL leads to five subsystems that can be solved exactly in time, and for which efficient and high-order methods in space and velocity can be used. 
As a consequence, the time accuracy of the resulting scheme only depends on the splitting error (which can be made arbitrarily small using high-order composition splittings \cite{hairer_lubich_wanner, yoshida}) and since the method is symplectic (as composition of symplectic flows), it maintains long term accuracy on invariants such as the total energy \cite{hairer_lubich_wanner}. Another interesting property that can be proven for the proposed scheme is the exact preservation of the norm of the ion spin $\|{\bf S}\|$. 

To validate the numerical results, we investigate the linearized VPLL system by deriving the pertinent dispersion relation, following \cite{Manfredi2019}. When the ion-electron coupling is turned off, the dispersion relation degenerates into 
the standard Bohm-Gross relation for plasmons and the magnon dispersion relation for the ion spins \cite{Eich:420521}. It is noteworthy that the typical plasmon timescale is about two orders of magnitude faster than that of magnons, which constitutes a considerable challenge for the numerical scheme.
In the case of Maxwell-Boltzmann equilibria, the dispersion relations can be solved numerically using dedicated libraries, e.g. Zeal \cite{zeal}. 
Moreover, analytical calculations are performed in the weak coupling regime.
Cross-validations between the roots of the dispersion relation and the results of the nonlinear code are performed and discussed. 

The rest of the paper is organized as follows. Section \ref{section:VPLL} lays the basis of the VPLL model equations and their nondimensional form. Section \ref{section:linear} discusses the linear response theory and the corresponding dispersion relation.
The numerical method is presented in section \ref{section:numericalmethod}. Results of numerical simulations are presented in section \ref{section:results}, both for a stable Maxwell-Boltzmann equilibrium and an unstable two-stream distribution function, and compared to linear-response results obtained from the dispersion relation, particularly for damping and growth rates. Conclusions are drawn in section \ref{section:conclusion}. Three Appendices provide some further details on the Maxwell-Boltzmann equilibrium with spin (Appendix \ref{appendix:equilibrium}), the dispersion relation (Appendix \ref{appendix:dispersion}), and the numerical splitting technique (Appendix \ref{split_app}).


\section{Vlasov-Poisson-Landau-Lifshitz model} \label{section:VPLL}

We consider a generic scenario where a magnetic material (e.g., nickel) is irradiated with a strong femtosecond laser pulse, so that some or most of the electrons are extracted from the bulk and can move freely, thus constituting a mobile electron plasma. The pulse heats up the electrons to a temperature equivalent to their Fermi energy, which for nickel is $E_F \approx 10 \,\rm eV$, while their density remains similar to that of the solid $n_e \approx 10^{29} \rm m^{-3}$. 
These parameters are close to those of the weakly degenerate plasmas typical of WDM \cite{Bonitz2020,Falk2018}.
During these initial instants, up to about 100~fs, the ions do not have time to move, and can thus be assimilated to an immobile, but magnetized, background. 

Within this broad context, our purpose here is to validate our numerical code, in the linear and nonlinear regimes, for parameters that are similar to those mentioned above. Hence, we will consider a one-dimensional (1D) model with periodic boundary conditions, and will investigate how a perturbed Maxwell-Boltzmann equilibrium evolves in time, for both the charge (plasmons) and spin (magnons) sectors. We will also analyze potentially unstable two-stream equilibria.

\subsection{Model equations}
The electrons are described by a four-component distribution function $(f_0, {\mathbf f})(t, {x}, {v})$ with ${\mathbf f} = (f_1, f_2, f_3)\in \mathbb{R}^3$, which is coupled to the continuous ion spin  distribution
${\mathbf S}(t, x) = (S_1, S_2, S_3)(t, x)$. The overall system of equations, for the space variable $x\in [0, L]\subset \mathbb{R}$ and velocity variable $v\in\mathbb{R}$,
is composed of set of kinetic equations for the electron distribution functions 
\cite{Manfredi2019,Crouseilles2023}
\begin{align}
\label{f_0}
\frac{\partial f_0}{\partial t} &+ {v}  \frac{\partial f_0}{\partial { x}} + \frac{e}{m} \frac{\partial V_H}{\partial { x}}  \frac{\partial f_0}{\partial { v}} - \frac{\mu_B }{m} 
\frac{\partial {\mathbf B}}{\partial x}\cdot  \frac{\partial {\mathbf f}}{\partial  v}= 0,\\
\label{f_i}
\frac{\partial f_\ell}{\partial t}  &+ {v} \frac{\partial f_\ell}{\partial { x}}+\frac{e}{m} \frac{\partial V_H}{\partial { x}}  \frac{\partial f_\ell}{\partial { v}}- \frac{\mu_B }{m}  \frac{\partial B_\ell}{\partial{  x}} \frac{\partial f_0}{\partial {v}} - \frac{e}{m}\left( {\mathbf B} \times {\mathbf f} \right)_\ell = 0, \ \ell = 1, 2, 3, 
\end{align}
and Landau-Lifshitz equation \cite{lakshmanan2011} for the ion spins 
\begin{equation}
\label{ion_S}
\frac{\partial {\mathbf S}}{\partial t} = \frac{a^2 J }{\hbar}\left({\mathbf S}^{} \times \frac{ \partial^2 {\mathbf S}^{}}{\partial x^2} \right) + \frac{K}{2\hbar} {\mathbf S}^{} \times \int {\mathbf f} \mathrm{d}{ v},
\end{equation}
where the first term on the right-hand side is the Heisenberg ion-ion magnetic exchange, whereas the second term represents the ion-electron magnetic exchange.

The scalar distribution function $f_0(t,x,v)$ represents, as usual, the probability to find an electron in the phase space volume located around $(x,v)$, at time $t$. Its moments yield the usual macroscopic quantities, such as the density $n_e(t,x) = \int f_0(t, x, v) \mathrm{d} v$.
In contrast, the vector distribution function $f_\ell(t,x,v)$ represents the mean spin polarization density of the electrons in the phase space volume located around $(x,v)$ at time $t$, along the $\ell$ direction. Its first moment ${\mathbf M}(t,x) = \int {\mathbf f}(t, x, v) \mathrm{d} v$ represents the electron spin density. For more details, see the recent review \cite{Manfredi2019}.
The relationship between this $(f_0, {\mathbf f})$ representation and the more standard representation as a $2 \times 2$ matrix with spin-up and spin-down components is also illustrated in the Appendix \ref{appendix:equilibrium}.

The self-consistent electric potential (Hartree potential) $V_H(t, x)$ obeys the Poisson equation
\begin{equation}
\label{poisson}
\epsilon_0 \, \partial_x^2 V_H= e\int f_0 \mathrm{d}{ v}- Ze\, n_{\rm ion}, 
\end{equation}
and the magnetic field appearing in \eqref{f_0}-\eqref{f_i} is primarily the one created by the ions
\begin{equation}
\label{eq:Bfield}
{\mathbf B}(t, x)=-\frac{K n_{\rm ion} {\mathbf S^{}(t, x)}}{2\mu_B} ,
\end{equation}
although external fields could also be considered.
Here, $e>0$ denotes the electron charge, $\hbar$ 
the Planck constant, $m$ the electron mass, $\epsilon_0$ the  permittivity of vacuum, $\mu_B=e\hbar/2m$ the Bohr magneton, $a$ the interatomic distance, $Z$ is the atomic number, $J$ and $K$ are respectively the ion-ion and electron-ion magnetic exchange constants, and $n_{\rm ion}$ is the fixed, homogeneous ion density. The full initial condition may be denoted as $(f_0, {\mathbf f}, V_H, \mathbf S)(t=0)=(f_0^{(0)}, {\mathbf f}^{(0)}, V_H^{(0)}, \mathbf S^{(0)})$, where $\varepsilon_0 \partial_x^2 V_H^0 =e \int f_0^0 \mathrm{d}{ v}-Ze n_{\rm ion}$.

Note how the $K$-terms couple the ion and electron spin dynamics: the magnetic field $\mathbf B$ given by \eqref{eq:Bfield} created by the ions acts on the spin part of the electron distribution functions $\mathbf f$ in \eqref{f_0}-\eqref{f_i}, while the electron spin density $\int {\mathbf f} \mathrm{d}{ v}$ acts on the LL equation \eqref{ion_S} for the ion spins.
A schematic view of the physical system under consideration is shown in Fig. \ref{fig:schematic}.
\begin{figure}[H]
    \centering
    \includegraphics[scale=0.5]{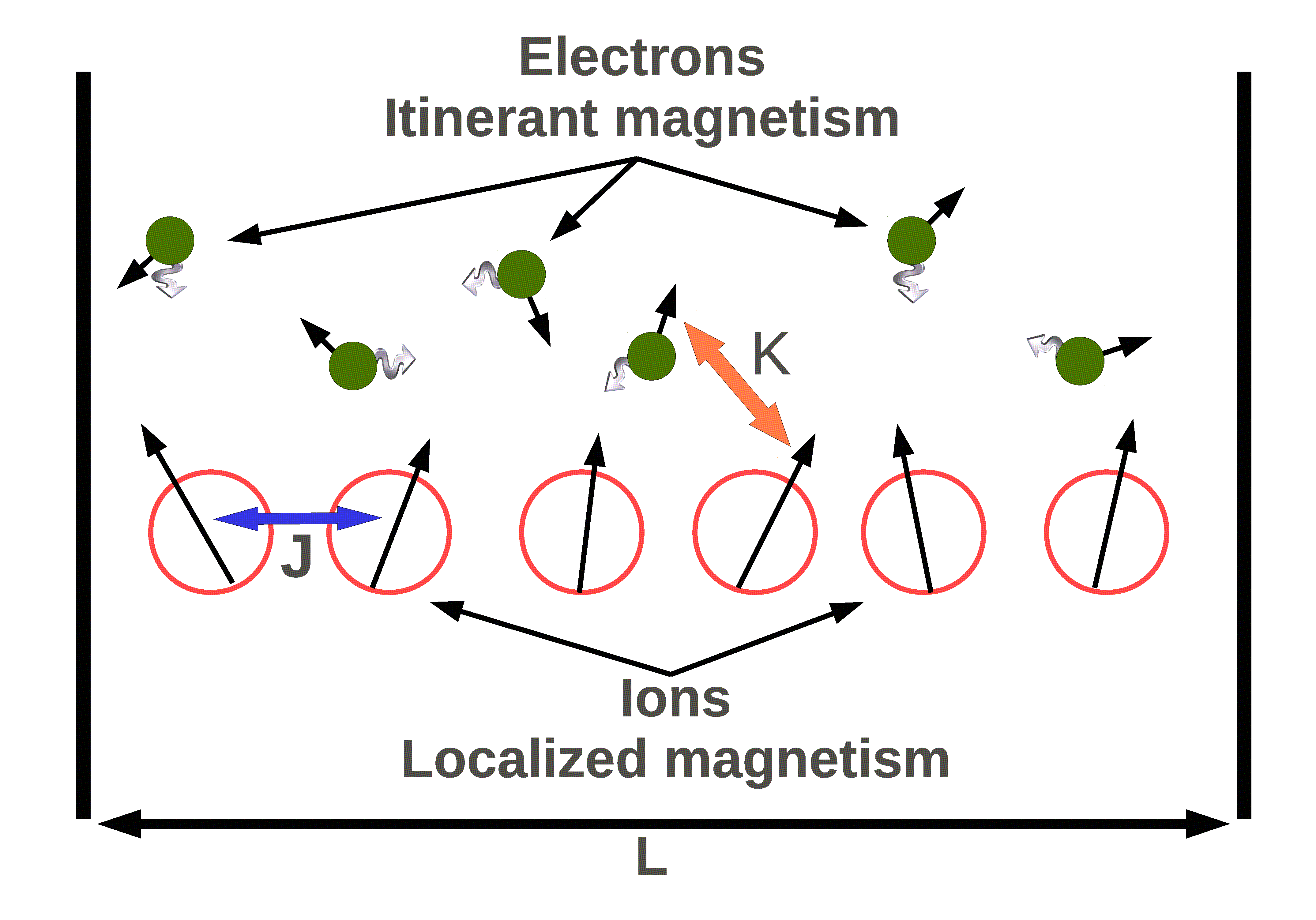}
    \caption{Schematic view of the physical system under consideration. The immobile ions (red circles) provide the main source of localized magnetism. They interact through magnetic exchange both with themselves (coupling constant $J$) and with the itinerant
    electrons, represented by green dots (coupling constant $K$).  }
    \label{fig:schematic}
\end{figure}

Form a mathematical viewpoint, the model \eqref{f_0}-\eqref{poisson} enjoys a Poisson structure with the following Hamiltonian functional
\begin{equation}
\label{hamiltonian}
\mathcal{H} = \frac{m}{2}\int {v}^2 f_0 \mathrm{d}{x} \mathrm{d}{v} + \mu_B \int {\mathbf f}  \cdot {\mathbf B} \mathrm{d}{x} \mathrm{d}{v} + \frac{\epsilon_0}{2} \int (\partial_x V_H)^2 \mathrm{d}{x} + \frac{a^2 J }{2} \int n_{\rm ion}
\sum_{\ell=1}^3  \Big(\frac{\partial {S^{}_\ell}}{\partial x}\Big)^2 \mathrm{d}{ x}.  
\end{equation}
Moreover, it is possible to construct a  Poisson bracket for two functionals 
$\mathcal{F}$ and $\mathcal{G}$ 
\begin{equation}
\label{poisson_bracket}
\begin{aligned}
\{\mathcal{F},\mathcal{G}\}=&\sum_{i=0}^3 \int \frac{f_0}{m} \left[ \frac{\delta \mathcal{F}}{\delta f_i}, \frac{\delta \mathcal{G}}{\delta f_i}\right]_{xv}  \mathrm{d}{ x}\mathrm{d}{ v} + \sum_{i=1}^3 \left( \int \frac{f_i}{m} \left[ \frac{\delta \mathcal{F}}{\delta f_0}, \frac{\delta \mathcal{G}}{\delta f_i}\right]_{xv} \mathrm{d}{ x}\mathrm{d}{ v} + \int \frac{f_i}{m} \left[ \frac{\delta \mathcal{F}}{\delta f_i}, \frac{\delta \mathcal{G}}{\delta f_0}\right]_{xv} \mathrm{d}{ x}\mathrm{d}{ v} \right) \\
& +\frac{e}{\mu_B m} \int {\mathbf f} \cdot \left( \frac{\delta \mathcal{F}}{\delta {\mathbf f}} \times  \frac{\delta \mathcal{G}}{\delta {\mathbf f}}\right) \mathrm{d}{ x}\mathrm{d}{ v} + \frac{1}{\hbar}\int \frac{{\mathbf S}^{}}{n_{\rm ion}} \cdot \left( \frac{\delta \mathcal{F}}{\delta {\mathbf S}^{}} \times  \frac{\delta \mathcal{G}}{\delta {\mathbf S}^{}}\right) \mathrm{d}{ x}. 
\end{aligned}
\end{equation}

\begin{remark}
It is easy to check that the bracket \eqref{poisson_bracket}
is bilinear, skew-symmetric, and satisfies Leibniz's rule, but it is not clear whether Jacobi's
identity is satisfied. Hence, this bracket is not strictly speaking a Poisson bracket; nevertheless we will still refer to it as a Poisson bracket for
the sake of simplicity. 
\end{remark}

With these notations in hand, the system \eqref{f_0}-\eqref{poisson} can be reformulated, after introducing the vector of unknowns $\mathcal{Z}=(f_0, {\mathbf f}, {\mathbf S}^{})\in\mathbb{R}^7$, as  
\begin{equation}
\label{eq:possionbracketequation}
\frac{\partial \mathcal{Z}}{\partial t} = \{ \mathcal{Z}, \mathcal{H} \}. 
\end{equation}

\subsection{Normalized dimensionless equations}
\label{subsec:normalized}
We rewrite the above equations \eqref{f_0}-\eqref{poisson} using dimensionless units that correspond to normalizing time to the inverse of the plasmon frequency $\omega_p=\sqrt{e^2n_e/\epsilon_0 m}$, velocities to the thermal speed $v_{th} = \sqrt{k_B T_e/m}$, and space to the Debye length $\lambda_D = v_{th}/\omega_p$, where $k_B$ is the Boltzmann constant. Hence the electric potential is normalized to $mv_{th}^2/e$, the electric field to $mv_{th}\omega_p/e$, and the magnetic field to $m \omega_p/e$.

Using these normalized units and defining the self-consistent electric field as $E_x=-\partial_x V_H$, the dimensionless kinetic equations read as (for simplicity of notation, we do not change the names of the dimensionless variables):
\begin{align}
\label{f_0_norm}
&\frac{\partial f_0}{\partial t} + {v}  \frac{\partial f_0}{\partial { x}} -E_x \frac{\partial f_0}{\partial { v}} - H  \frac{\partial {\mathbf B}}{\partial{ x}}\cdot \frac{\partial {\mathbf f}}{\partial { v}}= 0,\\
\label{f_i_norm}
&\frac{\partial f_\ell}{\partial t}  + {v} \frac{\partial f_\ell}{\partial { x}}-E_x  \frac{\partial f_\ell}{\partial { v}}- H  \frac{\partial B_\ell}{\partial{  x}} \frac{\partial f_0}{\partial {v}} - \left( {\mathbf B} \times {\mathbf f} \right)_\ell = 0, \;\; \ \ell = 1, 2, 3,
\end{align}
where
\begin{equation}
 {\mathbf B}=-\frac{\widetilde{K} {\mathbf S^{}}}{2}
 \label{eq:magn-field-ions}
\end{equation}
is the magnetic field created by the ions.

The dimensionless Planck constant
\begin{equation}
\displaystyle H=\frac{\hbar \omega_p}{2 mv_{th}^2}
\end{equation}
quantifies the relative importance of quantum effects with respect to thermal effects.
We also note that $H$ can be written in terms of the quantum coupling parameter $\Gamma_q = \hbar \omega_p/E_F$ and the degeneracy parameter $\Theta = T_e/T_F$ as: $H = \Gamma_q/(2\Theta)$.
In turn, the quantum coupling parameter is related to the Wigner-Seitz radius $r_s$ through the relationship \cite{Bonitz2020}: 
\begin{equation}
\label{eq:rs}
\Gamma_q^2 = \frac{r_s}{a_0} \, \frac{16}{9\pi} \left(\frac{12}{\pi}\right)^{1/3} \approx 0.88\, \frac{r_s}{a_0} ,
\end{equation}
where $a_0 = 4\pi \varepsilon_0 \hbar^2/(me^2)$ is the Bohr radius.

The normalized LL equation becomes 
\begin{align}
\label{S_norm}
&\frac{\partial {\mathbf S}^{}}{\partial t} = A\left({\mathbf S}^{} \times \frac{ \partial^2 {\mathbf S}^{}}{\partial x^2} \right) + Z\frac{\widetilde{K}}{4} {\mathbf S}^{} \times \int {\mathbf f} \mathrm{d}{ v},
\end{align}
with the dimensionless magnetic exchange constants written as $\displaystyle A=\frac{{a}^2}{\lambda_D^2}\, \frac{J}{\hbar \omega_p}$ and $\displaystyle \widetilde{K} = \frac{2 K n_{\rm ion}}{\hbar \omega_p}$. 
Finally, the dimensionless Poisson equation is 
\begin{align}
\label{poisson_norm}
&-\partial_x E_x=\int f_0 \mathrm{d}{ v}-1.  
\end{align}
The total energy in dimensionless units  is given by the Hamiltonian $\displaystyle \mathcal{H}=\mathcal{H}_v+\mathcal{H}_E+\sum_{i=1}^3\Big(\mathcal{H}_{Z,i}+\mathcal{H}_{spin, i}\Big)$, with
\begin{equation}
\begin{aligned}
\label{hamiltonian_dim}
\mathcal{H}_{v}& = \frac{1}{2}\int v^2 f_0 \mathrm{d}{ x}\mathrm{d}{v},\;\;\;\;\;\;\;\; 
\mathcal{H}_{E} = \frac{1}{2} \int \Big(\frac{\partial V_H}{\partial x}\Big)^2 \mathrm{d}{x},\\
\mathcal{H}_{Z,i} &= H \int   f_i B_i  \mathrm{d}x\mathrm{d}v, \;\; \mathcal{H}_{spin,i}= AH \int \Big(\frac{\partial {S^{}_i}}{\partial x}\Big)^2 \mathrm{d}{x},
\end{aligned}
\end{equation}
where the various terms correspond to the kinetic energy ($\mathcal{H}_{v}$), the Hartree electric energy ($\mathcal{H}_{E}$), the magnetic Zeeman energy ($\mathcal{H}_{Z}$), and the spin energy ($\mathcal{H}_{spin}$).

We consider an electron plasma in the WDM regime, with density $n_e = n_{\rm ion} = 9.17 \times 10^{28} \,\rm m^{-3}$ ($Z=1$) and temperature $k_B T_e = 16.58\,\rm eV$. This choice yields for the time, velocity, and length scales: $\omega_p=1.71\times 10^{16}\, \text{s}^{-1}$, $v_{th}=1.71 \times 10^6 \, \text{m  s}^{-1}$, and $\lambda_D=10^{-10}\, \text{m}$.
As to the dimensionless parameters, we find: normalized Planck constant $H=0.339$, quantum coupling parameter $\Gamma_q = 1.516$, Wigner-Seitz radius $r_s/a_0 = 2.60$ (corresponding to nickel), and degeneracy parameter $\Theta = 2.24$.

For the magnetic exchange coupling constants, we use values close to those of nickel \cite{Hurst2018}: $J = 0.022\,\rm eV$ and $K= 0.01 \,\rm eV \, nm^3$. Taking the lattice spacing $a = 2 r_s = 0.275 \,\rm nm$, this yields for the dimensionless parameters: $A=0.0148$ and $\widetilde{K}=0.161$.

\section{Linear analysis and dispersion relations}
\label{section:linear}
\subsection{Linear analysis for a generic equilibrium}
In order to validate the model \eqref{f_0_norm}-\eqref{poisson_norm} in the linear response regime, we perform a linear analysis to derive the pertinent dispersion relation. First, we 
start with the following homogeneous  stationary state: 
\begin{eqnarray*}
f_0^{(0)}=f_0^{(0)}(v), \; && \;  f_3^{(0)}=f_3^{(0)}(v), \nonumber\\ f_1^{(0)}=f_2^{(0)}=0, \; && \; S^{(0)}_1=S^{(0)}_2=E_x^{(0)}=0, \; \mbox{ and }  S^{(0)}_3=1, 
\end{eqnarray*}
where the superscript "$(0)$" stands for equilibrium.
This corresponds to an ion system that is fully polarized in the $S_3$ direction, and an electron system that is partially polarized in the same direction. The degree of electron spin polarization depends on the choice of $f_3^{(0)}(v)$, and can be characterized by a single number $\eta = \int_{-\infty}^{\infty} f_3^{(0)}(v) dv$, with $\eta \in [-1,1]$

We then derive the linearized system and study the propagation of a perturbation around the stationary state. We thus consider solutions in the form 
$$f_0=f_0^{(0)}+ f_0^{(1)},\, f_\ell=f_\ell^{(0)}+ f_\ell^{(1)},\, E_x =E_x^{(0)}+ E_x^{(1)},\, {\rm and}\,\, S^{}_\ell=S^{(0)}_\ell+ S_\ell^{(1)}.$$
Inserting these solutions into the system 
\eqref{f_0_norm}-\eqref{poisson_norm} and neglecting quadratic terms leads to the following linear system 
\begin{align}
&\frac{\partial  f_0^{(1)}}{\partial t} + {v}  \frac{\partial  f_0^{(1)}}{\partial { x}} - E_x^{(1)} \frac{\partial f_0^{(0)}}{\partial { v}} +\frac{H\widetilde{K}}{2} \frac{\partial  S^{(1)}_3}{\partial{ x}} \frac{\partial f_3^{(0)}}{\partial { v}}= 0,\\
&\frac{\partial f_1^{(1)}}{\partial t} + {v}  \frac{\partial  f_1^{(1)}}{\partial { x}} +\frac{H\widetilde{K}}{2} \frac{\partial  S^{(1)}_1}{\partial{ x}} \frac{\partial f_0^{(0)}}{\partial { v}} -\frac{\widetilde{K}}{2}( f_2^{(1)}-f_3^{(0)}  S^{(1)}_2)  = 0,\\
&\frac{\partial  f_2^{(1)}}{\partial t} + {v}  \frac{\partial  f_2^{(1)}}{\partial { x}} +\frac{H\widetilde{K}}{2} \frac{\partial S^{(1)}_2}{\partial{ x}} \frac{\partial f_0^{(0)}}{\partial { v}} +\frac{\widetilde{K}}{2}( f_1^{(1)}-f_3^{(0)} S^{(1)}_1)  = 0,\\
&\frac{\partial  f_3^{(1)}}{\partial t} + {v}  \frac{\partial f_3^{(1)}}{\partial { x}} - E_x^{(1)} \frac{\partial f_3^{(0)}}{\partial { v}} +\frac{H\widetilde{K}}{2} \frac{\partial  S^{(1)}_3}{\partial{ x}} \frac{\partial f_0^{(0)}}{\partial { v}}= 0,\\
&\frac{\partial  S_1^{(1)}}{\partial t}=-A\frac{ \partial^2  S^{(1)}_2}{\partial x^2}+\frac{\widetilde{K}}{4}\Big( S^{(1)}_2\int f_3^{(0)} \mathrm{d}{ v} -\int  f_2^{(1)} \mathrm{d}{ v}\Big),\\
&\frac{\partial  S_2^{(1)}}{\partial t}=A \frac{ \partial^2  S^{(1)}_1}{\partial x^2}-\frac{\widetilde{K}}{4}\Big( S^{(1)}_1\int f_3^{(0)} \mathrm{d}{ v} -\int  f_1^{(1)} \mathrm{d}{ v}\Big),\\
&\frac{\partial S^{(1)}_3}{\partial t}=0,\\
&-\partial_x  E_x^{(1)}=\int  f_0^{(1)} \mathrm{d}{ v}.
\end{align}
By performing Fourier (in space) and Laplace (in time) transforms of the above linear system of equations, we can  derive an equation relating the frequency $\omega$ and the wave number $k$ (we shall further refer to $\omega_e$ for the charge branch of the dispersion relation and $\omega_s$ for the spin branch). Since $S_3$ does not depend on time, the dispersion relation for $f_0^{(1)}$ and $E_x^{(1)}$ is the same as the standard Bohm-Gross relation for unpolarized electrons, that is  
\begin{equation}
\label{dispersion_f}
D_e(\omega_e,k) \equiv 1-\frac{1}{k}\int \frac{ \partial_v f_0^{(0)}}{kv-\omega_e}\,\mathrm{d}{ v}=0
\end{equation}
(here and in the following, velocity integrals are understood as being from $-\infty$ to $+\infty$).
Hence, at the level of the linear response, the spin and charge motions are completely separated. This is an important fact, as it means that an excitation (e.g., a laser pulse) acting only on the charge density will not trigger any response in the spin dynamics.
In order to generate a spin dynamics, one needs either a strong pulse that generates nonlinear effects, or an excitation that acts directly on the spins (e.g., via the magnetic part of the laser pulse).

Next, we consider the equations for $f_1^{(1)}$, $f_2^{(1)}$, $S_1^{(1)}$ and $S^{(1)}_2$, 
which lead to the dispersion relation for the ion spin motion:
\begin{equation}
\label{dispersion_S}
D_S(\omega_s,k) \equiv -\Big[\omega_s - \frac{\widetilde{K}^2}{8}\Big(\frac{\widetilde{K}H k}{2}  I_0+I_1\Big)\Big]^2+\Big[A k^2+\frac{\widetilde{K}\eta}{4}- \frac{\widetilde{K}^2}{8}\Big(H k  I_3+\frac{\widetilde{K}}{2} I_2\Big)\Big]^2 = 0,
\end{equation}
where we have defined the integrals 
\begin{eqnarray*}
I_0&=&\int \frac{\partial_v f_0^{(0)}}{(\widetilde{K}/2)^2-(vk-\omega_s)^2} \mathrm{d}{ v}, \;\;\;\;\;\;\;\;\;\;\;\;\;\;\;\;\;\;\;\;\;  I_1=\int \frac{(vk-\omega_s) f_3^{(0)}}{(\widetilde{K}/2)^2-(vk-\omega_s)^2} \mathrm{d}{ v},\\
I_2&=&\int \frac{ f_3^{(0)}}{(\widetilde{K}/2)^2-(vk-\omega_s)^2} \mathrm{d}{ v},  \;\;\;\;\;\;\;\;\;\;\;\;\;\;\;\;\;\;\;\;\; I_3=\int \frac{(vk-\omega_s)\partial_v f_0^{(0)}}{(\widetilde{K}/2)^2-(vk-\omega_s)^2} \mathrm{d}{ v}. 
\end{eqnarray*}

Note that, when one neglects the electron-ion coupling, i.e. $\widetilde{K}=0$, the spin branch of the dispersion relation reduces to: $\omega_s = \pm A k^2$, which is the standard magnon dispersion relation \cite{Ashcroft}. 
In contrast, the dispersion relation for the electrons yields, from \eqref{dispersion_f}, $\omega_e \approx \omega_p$. Taking the ratio of the magnon and plasmon frequencies yields:
\begin{equation} \label{eq:scalediffer}
\frac{\omega_s}{\omega_e} = \frac{A k^2}{\omega_p} \approx 8.6 \times 10^{-3},
\end{equation}
where we used the parameters given in section \ref{subsec:normalized}, i.e., $\omega_p=1.71\times 10^{16}\, \text{s}^{-1}$ and $A=0.0148$, and considered a typical length $k^{-1} = 10\,\rm nm$. This indicates that the timescale of magnons is about two order of magnitudes slower than that of plasmons. This fact has an obvious impact on the numerical simulations, as many hundreds of plasmon cycles have to be resolved before one can observe a sizeable response in the ion spins.

\subsection{Maxwell-Boltzmann equilibrium}
Now we assume the stationary states $f_0^{(0)}, f_\ell^{(0)}$ to be Gaussian functions, so that  $I_0, I_1, I_2,I_3$ can be expressed using the Fried-Comte function \cite{Fried1961} $Z(z)=\frac{1}{\sqrt{\pi}}\int_{\mathbb{R}} \frac{e^{-t^2}}{t-z}\mathrm{d}{ t}, \; z\in \mathbb{C}$, which can itself be expressed  using the erfi function  $\mbox{erfi}(z)=\frac{2}{\sqrt{\pi}}\int_0^z e^{t^2}\mathrm{d}{ t}, \; z\in \mathbb{C}$ and is tabulated in several scientific libraries. 

Let consider that the following homogeneous equilibrium 
\begin{equation}
f_0^{(0)}(v)=\frac{1}{\sqrt{\pi}}e^{-v^2}, \;\;\;\;\;\;\;\;\;\;   f_3^{(0)}(v)=\eta \frac{1}{\sqrt{\pi}}e^{-{v^2}}, 
\label{eq:MBequilibrium}
\end{equation}
where  $\eta=\int f_3^{(0)} dv$ is the spin polarization rate of the electrons (see Appendix \ref{appendix:equilibrium} for further details). The dispersion function $D_e$ for the charge dynamics becomes 
$$
D_e(\omega_e,k)=1+\frac{2}{k^2}\Big[1+\frac{\omega_e}{k}Z\Big(\frac{\omega_e}{k}\Big)\Big]. 
$$
while the spin dispersion function $D_S$ is 
\begin{eqnarray}
D_{S}(\omega_s,k)&=& -\Big\{\omega_s + \frac{c_0}{k} \Big[Z\Big(\frac{\omega_s+\widetilde{K}/2}{k}\Big)+Z\Big(\frac{\omega_s-\widetilde{K}/2}{k}\Big)\Big]-c_1 \Big[Z'\Big(\frac{\omega_s-\widetilde{K}/2}{k}\Big)-Z'\Big(\frac{\omega_s+\widetilde{K}/2}{k}\Big)\Big]\Big\}^2\nonumber\\
&&\hspace{-2.2cm}+\Big\{A k^2+d +\frac{c_0}{k}  \Big[Z\Big(\frac{\omega_s+\widetilde{K}/2}{k}\Big)-Z\Big(\frac{\omega_s-\widetilde{K}/2}{k}\Big)\Big]+ c_1 \Big[Z'\Big(\frac{\omega_s-\widetilde{K}/2}{k}\Big)+Z'\Big(\frac{\omega_s+\widetilde{K}/2}{k}\Big)\Big]\Big\}^2
\label{dispersion}
\end{eqnarray}
with $c_0=\widetilde{K}^2\eta/16$,  $c_1=\widetilde{K}^2 H/16,$ and $d=\widetilde{K}\eta/4 $. 
Moreover, the complex-valued function $Z$ and its derivative are given by 
\begin{equation*}
    Z(z) = \sqrt{\pi}\exp(-z^2)(i - \mbox{erfi}(z)), \;\; Z'(z)=-2(z Z(z)+1).
\end{equation*}


\subsection{Analysis and computation of the spin dispersion relation}
\label{subsec:dispersion-comp}

In this section, we will use another form of the dispersion function which is strictly equivalent  to $D_S$ given by \eqref{dispersion} . $D_S$ can be written as the product of two different functions $D_S=D_-D_+$ (see Appendix \ref{appendix:compact_dispersion_relation} for further details), each of which generates the same solutions, but with different signs. 
In the following, we consider the function that gives rise to positive real frequencies in the limiting case $\widetilde{K}=0$, i.e.
\begin{equation}
    D_{-}(\omega_s,k)=
    \omega_s
    -Ak^2
    -\frac{\widetilde{K}}{4}\int  f_3^{(0)}\mathrm{d}v
    +\frac{\widetilde{K}^2}{8k}
    \int  \frac{f_3^{(0)}}{v-\left(\frac{\omega_s-\widetilde{K}/2}{k}\right)}\mathrm{d}v
    -\frac{H\widetilde{K}^2}{8}
    \int  \frac{\partial_v f_0^{(0)}}{v-\left(\frac{\omega_s-\widetilde{K}/2}{k}\right)}\mathrm{d}v,
     \label{eq:dispertion_analytical_form_MB_with_distribution}
\end{equation}
or, in terms of the plasma dispersion function $Z$,
\begin{equation}
    D_{-}(\omega_s,k)=
    \omega_s
    -Ak^2
    -\frac{\widetilde{K}\eta}{4}
    +\frac{\widetilde{K}^2\eta}{8k}
    Z\left(\frac{\omega_s-\widetilde{K}/2}{k}\right)
    -\frac{H\widetilde{K}^2}{8}
    Z'\left(\frac{\omega_s-\widetilde{K}/2}{k}\right) .
    \label{eq:dispertion_analytical_form_MB}
\end{equation}

This formulation highlights the different contributions to the magnon frequency. Let us spell out each term of the right-hand side of \eqref{eq:dispertion_analytical_form_MB_with_distribution}:
\begin{itemize}
    \item The first two terms yield the standard dispersion relation for magnons, $\omega_s = A k^2$;
    \item The next term shifts the magnon frequency due to ion precession around the magnetic field generated by electronic spins at steady state;
    \item The last two terms introduce corrections that are brought over by electrons that possess specific (resonant) velocities, either in their spin distribution $f_3^{(0)}$ or their charge distribution $f_0^{(0)}$ at equilibrium. This is similar to the resonant electrons that are responsible for Landau damping in spin-less plasmas.  
\end{itemize}

Equation \eqref{eq:dispertion_analytical_form_MB} possesses complex solutions in $\omega_s$, due to the complex-valued function $Z$. Physically, this means that some resonances occur in the electron population when the velocity is equal to (restoring physical dimensions for clarity) $v =\frac{\omega_s}{k}-\frac{\omega_L}{k}$, where $\omega_L=eB/m= 2\mu_B B/\hbar$ is the Larmor frequency of an electron spin in the magnetic field created by the (fully polarized) ions, $B=K n_{\rm ion}/(2\mu_B)$.
Thus, $\omega_s/k \equiv v_s$ is the phase velocity of the ion spin wave (the magnon), whereas  $\omega_L/k\equiv v_L$ is the phase velocity of the electronic spin wave propagating in the magnetized environment created by the polarized ions.
The resonance occurs when the electron spin precesses at the same frequency as the magnon, shifted by Doppler effect due to the electron velocity with respect to the fixed ions. In terms of the phase velocities, this can be written as: $v_s -v=v_L$.

This resonance behaves similarly to the Electron Cyclotron Resonance Heating (ECRH) effect in fusion plasmas, with two major differences. First, the ion spin wave (magnon) plays the role of the external electromagnetic wave in ECRH; second, the magnetic moment of the electrons is not orbital as in ECRH, but instead is due to the electron's intrinsic spin.

It is useful to compute the dispersion function $D_-(\omega_s,\widetilde{K})$ in terms of the coupling constant $\widetilde{K}$ and the frequency $\omega_s$, for a fixed value of the wave number $k$. Then, the solutions of the dispersion relation can be computed along a path in the $(\omega_s,\widetilde{K}$) plane, by solving the equation
\begin{align}
    \partial_{\widetilde{K}}D_- d\widetilde{K}+\partial_{\omega_s}D_- d\omega_s=0,
\end{align}
starting from known solutions, for instance the one at zero coupling $\omega_s(\widetilde{K}=0) \equiv \omega_0 =Ak^2$.
Solving for $\omega_s(\widetilde{K})$ yields
\begin{equation}
    \omega_s(\widetilde{K})=\omega_0-\int_0^{\widetilde{K}}  \left.\frac{\partial_{\widetilde{K}}D_-}{\partial_{\omega_s}D_-}\right|_{\omega_s(\widetilde{K}),\widetilde{K}} d\widetilde{K} .
    \label{eq:dispertion_integral_form_MB}
\end{equation}
Numerically, the solution is found by starting at $\widetilde{K}=0$ and then increasing $\widetilde{K}$ of small steps $d\widetilde{K}$ until the desired value is reached. The derivatives of $D_{-}$ used in \eqref{eq:dispertion_integral_form_MB} are given in the Appendix \ref{appendix:D_differential}.

In the figures \ref{fig:dispersion-selfconsistent} and \ref{fig:dispersion-eta}, we show the results obtained from equation \eqref{eq:dispertion_integral_form_MB} for three cases with same wave number $k=0.5$, but  different electron spin polarization $\eta$. 
The results of the dispersion relation are compared to numerical results obtained with the fully nonlinear code with a small perturbation around the equilibrium, as detailed in section \ref{section:results}.
For all cases, the agreement is excellent, which constitutes a cross-validation for both the numerical code and the above analytical developments.

In figure \ref{fig:dispersion-selfconsistent}, we use the value of $\eta$ that is consistent with electrons at thermal equilibrium that are polarized by the magnetic field $B$ created by the magnetized ions, see equation \eqref{eq:Bfield} (we shall refer to this case as the "self-consistent" case). In this case, the spin polarization is given by $\eta= \tanh(2\mu_B B /k_B T_e) = \tanh(H\widetilde{K})$  and obviously depends on the electron-ion magnetic coupling -- more details are given in Appendix \ref{appendix:equilibrium}.

In contrast, in figure \ref{fig:dispersion-eta}, we use two arbitrary values of the electron spin polarization, $\eta=0.5$ and $\eta=-0.5$. The negative value means that the electrons are polarized in the opposite direction with respect to that of the self-consistent case. These values might be obtained through an external magnetic field that pre-polarizes the electrons prior to the application of a small perturbation. Nevertheless, one should keep in mind that, to achieve such large spin polarizations, a very strong magnetic field would be needed, of the order of several hundred teslas. 

For these values of $\eta$, the imaginary part of $\omega_s$ is significantly different from zero. In particular, for $\eta=0.5$ there is a damping of the perturbation (${\rm Im}\, \omega <0$), whereas for $\eta=-0.5$ we observe an  instability (${\rm Im} \, \omega >0$) .
This behaviour can be interpreted as follows. When $\eta>H\widetilde{K}>0$, the electron polarization has the same direction as in the  self-consistent case, hence the perturbation is damped, as the system tries to return to a state that has the "natural" direction of polarization.  In contrast, when $\eta<H\widetilde{K}$ (and, in particular, when $\eta$ is negative) the system becomes unstable in an attempt to restore the "correct" direction of polarization.
When the value of $\eta$ corresponds to the self-consistent case, as in figure \ref{fig:dispersion-selfconsistent}, the system is marginally stable (${\rm Im} \, \omega \approx 0$). 
{Interestingly, in the self-consistent case the first-order correction in the electron-magnon coupling $\widetilde{K}$ disappears, see equation \eqref{eq:dispersion_weak_coupling_MB}. Hence, figure \ref{fig:dispersion-selfconsistent} shows almost no variation of the real and imaginary parts of the magnon frequency for low values of $\widetilde{K}$.

\begin{figure}[!ht]
    \centering
    \includegraphics[scale=0.4]{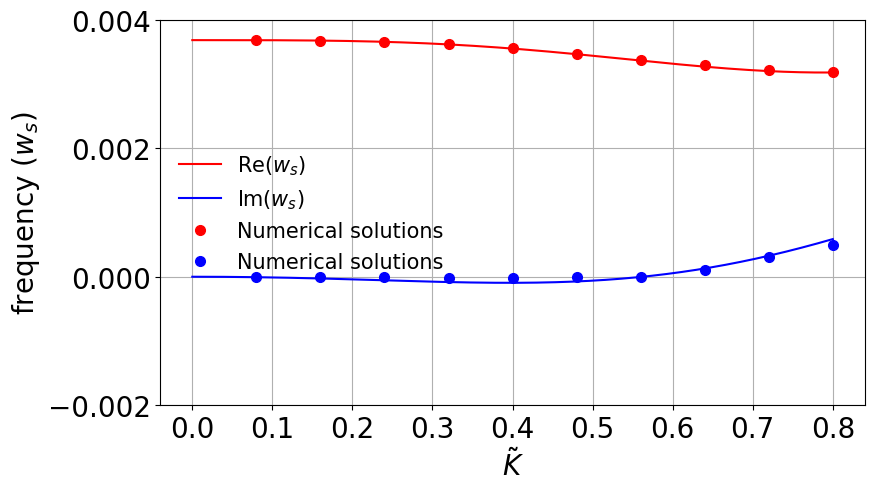}
    \caption{Magnon frequency $\omega_s$ for different normalized magnetic coupling constants $\widetilde{K}$, obtained from equation \eqref{eq:dispertion_integral_form_MB} (continuous lines, red for the real part of the frequency, blue for the imaginary part), for $k=0.5$ and $\eta=\tanh(H\widetilde{K})$. Note that the electron spin polarization $\eta$ is different for different values of $\widetilde{K}$. The dots represent numerical results obtained with the full numerical code described in the forthcoming sections. Note that, for this self-consistent case, the imaginary part remains very small with respect to the real part of the frequency.
 }
    \label{fig:dispersion-selfconsistent}
\end{figure}

\begin{figure}
\center{
\subfigure[]{\includegraphics[scale=0.35]{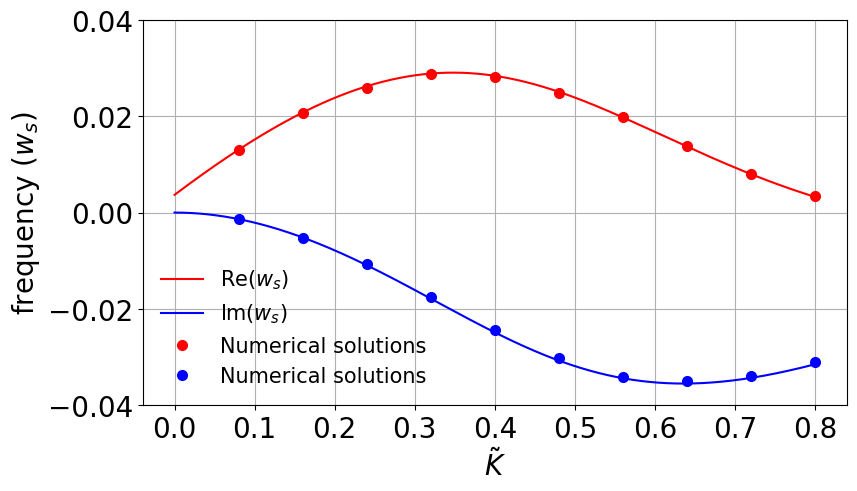}}
\subfigure[]{\includegraphics[scale=0.35]{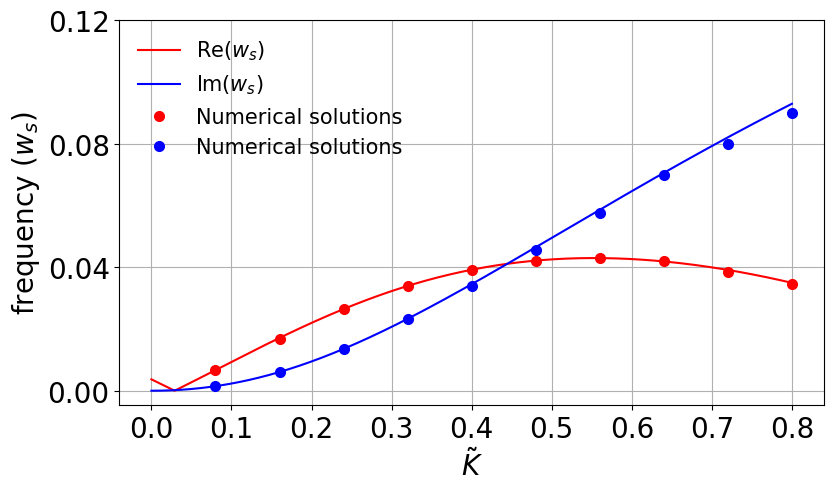}}
}
\caption{Magnon frequency $\omega_s$ for different normalized magnetic coupling constants $\widetilde{K}$, obtained from equation \eqref{eq:dispertion_integral_form_MB} (continuous lines, red for the real part of the frequency, blue for the imaginary part), for wave number $k=0.5$, and electron polarizations $\eta=0.5$ (a) and $\eta=-0.5 $ (b). The dots represent numerical results obtained with the full numerical code described in the forthcoming sections. According to the value of $\eta$, the system is either stable (a) or unstable (b). }
\label{fig:dispersion-eta}
\end{figure}

\subsection{Weak coupling regime}
\label{subsec:dispersion-weak}

From equation \eqref{eq:dispertion_analytical_form_MB}, the ion spin dispersion relation can be written as
\begin{equation}
    \omega_s = Ak^2+\frac{\widetilde{K}}{4}\left(\eta-H\widetilde{K}\right)-Z\left(\frac{\omega_s-\omega_L}{k}\right)\left[\frac{\widetilde{K}^2\eta}{8k}+\frac{\widetilde{K}^2H}{4}\left(\frac{\omega_s-\omega_L}{k}\right)\right]  \equiv G(\omega_s).
    \label{eq:iterations}
\end{equation}
This is a transcendental equation for $\omega_s$, which cannot be solved exactly, except numerically as was done in the preceding subsection.
An approximate solution to \eqref{eq:iterations} can be obtained iteratively, by starting with the solution for zero coupling, $\omega_0 = Ak^2$, then inserting this solution into the right-hand side of \eqref{eq:iterations}, which yields
\begin{equation}
    \omega_s \approx \omega_0+\frac{\widetilde{K}}{4}\left(\eta-H\widetilde{K}\right)-Z\left(\frac{\omega_0-\omega_L}{k}\right)\left[\frac{\widetilde{K}^2\eta}{8k}+\frac{\widetilde{K}^2H}{4}\left(\frac{\omega_0-\omega_L}{k}\right)\right]
  \label{eq:iteration-1}
\end{equation}
which is valid for weak coupling  $\widetilde{K}\ll1$.
This procedure can be recast as a fixed-point problem: $\omega_s^{(\ell+1)} = G(\omega_s^{(\ell)}), \,\ell\in \mathbb{N}$, with $\omega_s^{(0)} = \omega_0=Ak^2$, to obtain second- and higher-order approximations.

As the value of the dimensionless coupling constant is indeed small, $\widetilde{K} \approx 0.16$, this weak-coupling approximation  should hold for most cases of interest. 
Since $\widetilde{K}/2 \equiv \omega_L / \omega_p$, physically this approximation means that 
the electron Larmor frequency is much smaller than the plasmon frequency, specifically here: $\omega_L \approx 0.08\, \omega_p$. If we add the fact that the magnon frequency is about $\omega_0 \approx 0.008 \,\omega_p$, see equation \eqref{eq:scalediffer},  we obtain the following scaling between the three timescales that are present in this problem:
$\omega_0 \ll \omega_L \ll \omega_p $.

Under such weak-coupling approximation, \eqref{eq:iteration-1} simplifies to (restoring physical dimensions):
\begin{equation}
    \omega_{s}=\omega_0+\frac{\omega_L}{2}\left(\eta-H\widetilde{K}\right)
    \left[
    1+ \frac{2\omega_L}{v_{th} k}D\left(-\frac{\omega_L}{k}\right)
    -i\sqrt{\pi} \,\frac{\omega_L}{v_{th} k} \exp\left(-\frac{\omega_L^2}{v_{th}^2 k^2}\right) \right] ,
    \label{eq:dispersion_weak_coupling_MB}
\end{equation}
where we used the fact that ${\rm Im}\, Z(x)=\sqrt{\pi}e^{-x^2}$ ($x\in\mathbb{R}$)
when evaluated on the real axis (i.e., $x\in\mathbb{R}$) \cite{Fried1961} and where $D$ is the Dawson function.
By looking at the imaginary part of $\omega_{s}$, two regimes clearly appear. If 
$\eta < H\widetilde{K}$, the imaginary part is positive, so that the magnetic perturbation is unstable and grows exponentially until the nonlinear regime is reached. If $\eta > H\widetilde{K}$, then the perturbation is damped and disappears after a few oscillations. 
Interestingly, the value of $\eta$ that discriminates between these two regimes, i.e., $\eta = H\widetilde{K}$, is precisely the value that corresponds to the self-consistent case, $\eta = \tanh(H\widetilde{K})$, in the approximation where $\widetilde{K} \ll 1$.

{The form of the spin dispersion relation \eqref{eq:dispersion_weak_coupling_MB} reveals that all the magnetic terms in the Vlasov model \eqref{f_0_norm}-\eqref{f_i_norm} are important and cannot be neglected: the Zeeman terms proportional to $H$, the electron precession term proportional to $\mathbf{B}$ and hence to $\widetilde{K}$, as well as the initial electron spin polarization $\eta$. The subtle interplay between these terms determines the stable or unstable nature of the linear response.
In contrast, as we have seen, the electric charge response is completely decoupled from the spin response, at least in the linear regime. Hence, one could neglect the electric field terms in \eqref{f_0_norm}-\eqref{f_i_norm} (or set the initial electric perturbation to zero) and the spin response would remain unchanged. However, the plasmon oscillations would be lost.
}

The results for both the exact dispersion relation \eqref{eq:dispertion_integral_form_MB} and the approximate formula \eqref{eq:dispersion_weak_coupling_MB} are shown in figure \ref{fig:Disp_MB_K_t.png} for a self-consistent case.
As expected, the agreement is good for values up to $\widetilde{K} \approx 1$, which cover most realistic values of the coupling constant.
Finally, from \eqref{eq:dispersion_weak_coupling_MB}, one can compute the maximum imaginary part with the parameters used in figure \ref{fig:Disp_MB_K_t.png}. Since $\tanh(H\widetilde{K})= H\widetilde{K} - (H\widetilde{K})^3/3+ \mathcal{O}((H\widetilde{K})^5)$, the imaginary part of $\omega_s$ is proportional to $\widetilde{K}^5 e^{-(\widetilde{K}/2k)^2}$. The maximum is then reached for $\widetilde{K}=\sqrt{10}\,k \approx 1.58$, which is also in agreement with the exact dispersion relation.

Finally, in figure \ref{fig:Disp_MB_k} we show the dependence of the magnon frequency on the wavenumber $k$, comparing the full dispersion relation with its first order \eqref{eq:iteration-1} and second order approximations.

\begin{figure}[H]
    \centering
    \includegraphics[scale=0.6]{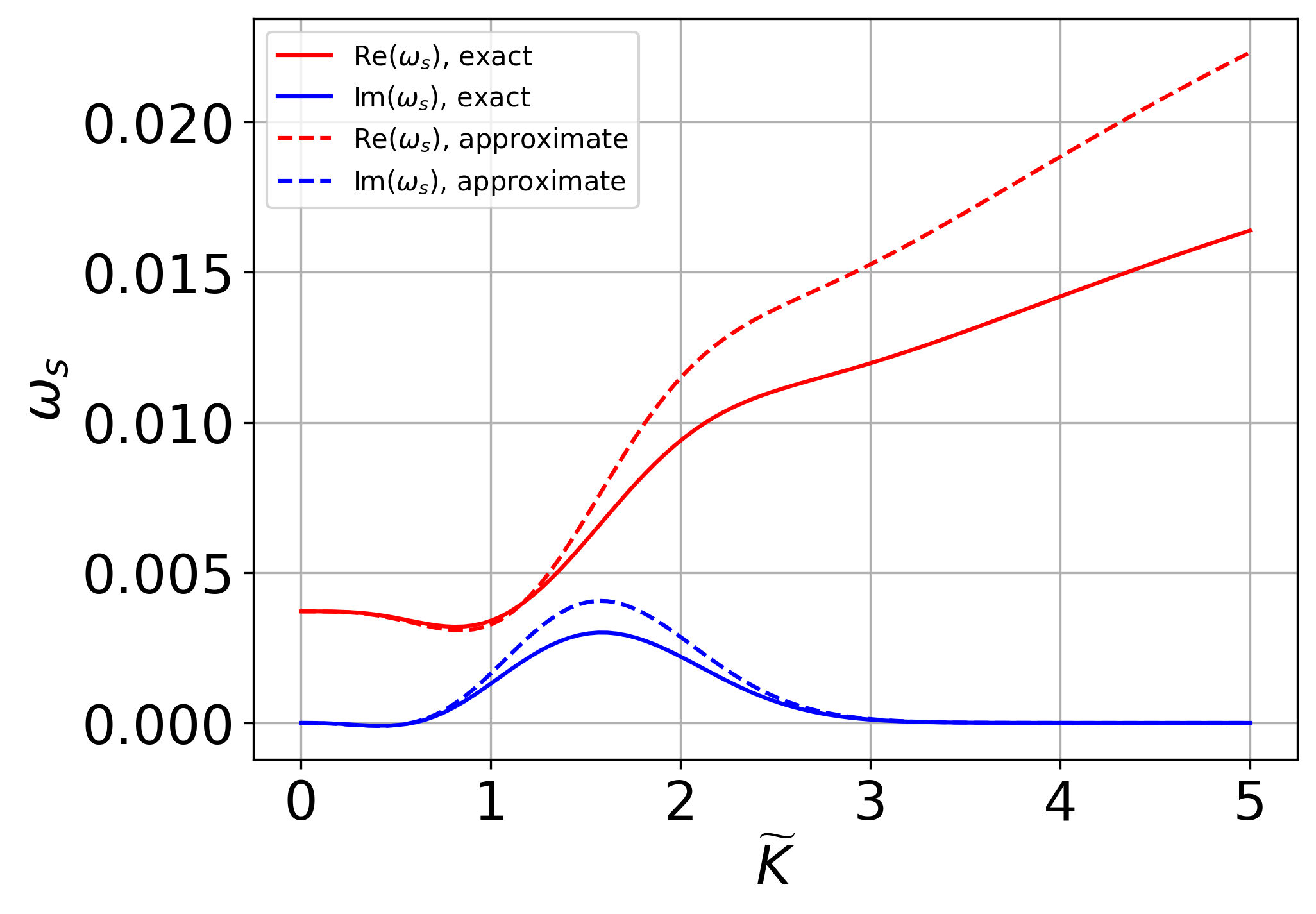}
    \caption{Magnon frequency $\omega_s$ as a function of the magnetic coupling constant $\widetilde{K}$, for the self-consistent case $\eta = \tanh(H\widetilde{K})$, and wavenumber $k=0.5$. 
    The solid lines represent the full dispersion relation computed numerically using  \eqref{eq:dispertion_integral_form_MB}, while the dashed lines are obtained with the simplified relation \eqref{eq:iteration-1}. Red lines refer to the real part of $\omega_s$, whereas blue lines refer to the imaginary part.
 }
    \label{fig:Disp_MB_K_t.png}
\end{figure}

\begin{figure}[H]
    \centering
    \includegraphics[scale=0.6]{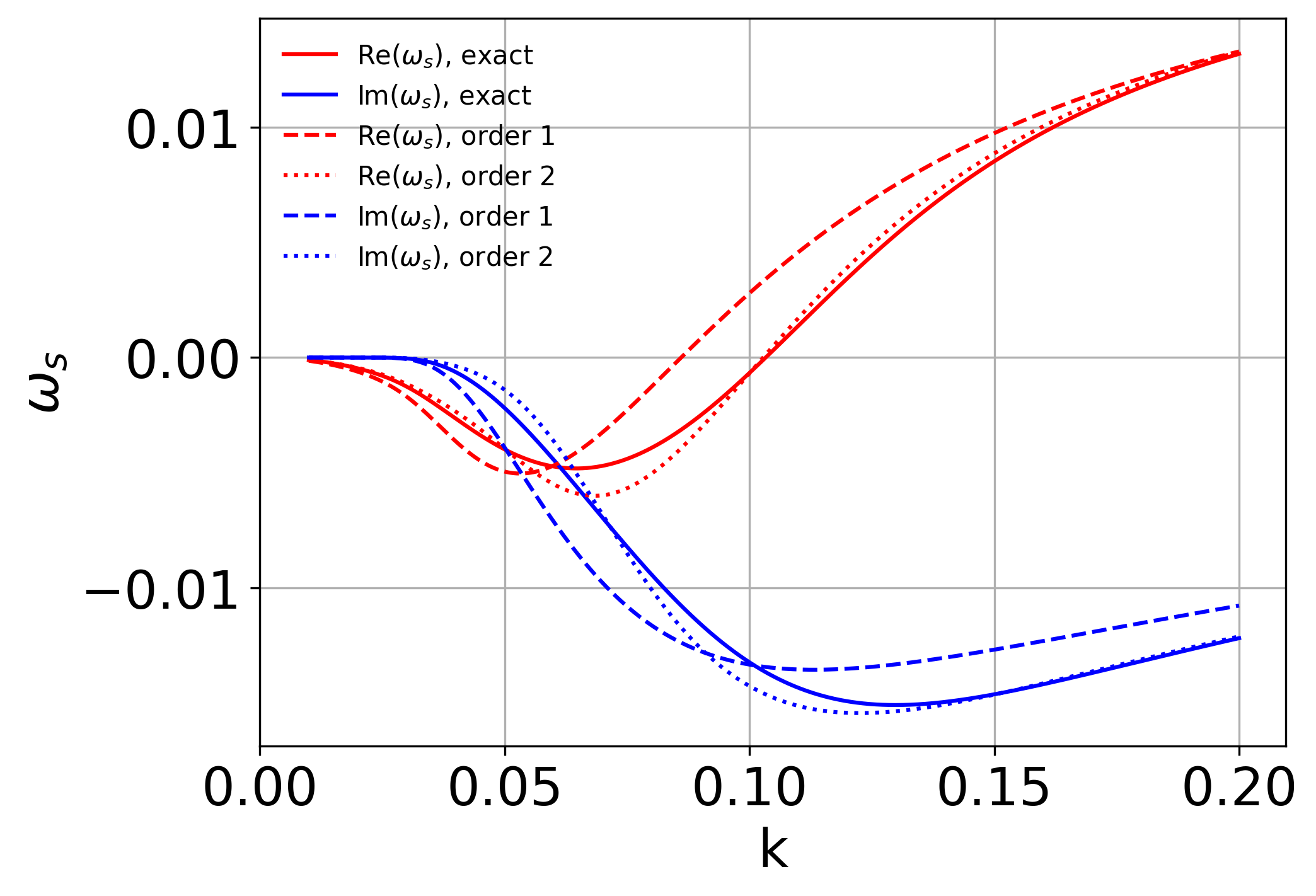}
    \caption{Magnon frequency $\omega_s$ as a function of the magnon wavenumber $k$, for  electron polarization $\eta = 0.5$, and magnetic coupling constant $\widetilde{K}=0.16$. 
    The solid lines represent the full dispersion relation computed numerically using  \eqref{eq:dispertion_integral_form_MB} where the integral and the derivative are not with respect to $\widetilde{K}$ but $k$. The other lines refer to the approximate linear theory obtained from \eqref{eq:iterations} at first order (dashed lines, given explicitely by equation \eqref{eq:iteration-1}) and second order (dotted lines).    
    Red lines refer to the real part of $\omega_s$, whereas blue lines refer to its imaginary part.
 }
    \label{fig:Disp_MB_k}
\end{figure}


\section{Numerical method}
\label{section:numericalmethod}

In this section, we present the numerical method used to solve the system of equations \eqref{f_0}-\eqref{poisson}. 
The method is based on a Hamiltonian splitting technique, together with a phase space discretization that uses Fourier spectral approximation for the space variable $x$ and finite volumes (PSM) for the velocity direction $v$, as in \cite{Crouseilles2023,laser2020}. 

The Hamiltonian can be split into five parts:
\begin{equation}\label{eq:splitHamiltonianparts}
\mathcal{H}=\mathcal{H}_{v}+\mathcal{H}_{E}+\mathcal{H}_{S_1}+\mathcal{H}_{S_2}+\mathcal{H}_{S_3},
\end{equation}
where 
\begin{equation}
\label{eq:splitHamiltonianparts_def}
\begin{aligned}
\mathcal{H}_{v}& = \frac{1}{2}\int v^2 f_0 \mathrm{d}{ x}\mathrm{d}{v},\quad
\mathcal{H}_{E} = \frac{1}{2} \int \Big(\frac{\partial V_H}{\partial x}\Big)^2 \mathrm{d}{x},\\
\mathcal{H}_{S_i} &= H \int   f_i B_i  \mathrm{d}x\mathrm{d}v+ AH \int \Big(\frac{\partial {S^{}_i}}{\partial x}\Big)^2 \mathrm{d}{x}, \;\; i=1, 2, 3.
\end{aligned}
\end{equation}
Let us remark that in this decomposition, $\mathcal{H}_{S_i} = \mathcal{H}_{Z,i} + \mathcal{H}_{spin,i}$ where the Zeeman energy  $\mathcal{H}_{Z,i}$ and the spin energy $\mathcal{H}_{spin,i}$ are given by  \eqref{hamiltonian_dim}. 
According to the  Hamiltonian splitting, we get from \eqref{eq:possionbracketequation}:
\begin{equation}
\label{full_model}
\frac{\partial \mathcal{Z}}{\partial t} = \{ \mathcal{Z}, \mathcal{H}_v \}+\{ \mathcal{Z}, \mathcal{H}_E \}+\{ \mathcal{Z}, \mathcal{H}_{S_1} \}+\{ \mathcal{Z}, \mathcal{H}_{S_2} \}+\{ \mathcal{Z}, \mathcal{H}_{S_3} \},
\end{equation}
which induces the five subsystems
\begin{equation}
\frac{\partial \mathcal{Z}}{\partial t} = \{ \mathcal{Z}, \mathcal{H}_v \},\ \frac{\partial \mathcal{Z}}{\partial t}=\{ \mathcal{Z}, \mathcal{H}_E \},\ \frac{\partial \mathcal{Z}}{\partial t}=\{ \mathcal{Z}, \mathcal{H}_{S_1} \},\ 
\frac{\partial \mathcal{Z}}{\partial t}=\{ \mathcal{Z}, \mathcal{H}_{S_2} \},\ 
\frac{\partial \mathcal{Z}}{\partial t}=\{ \mathcal{Z}, \mathcal{H}_{S_3} \}. 
\end{equation}
As detailed in the Appendix \ref{split_app}, each subsystem can be solved exactly, which means that the error in time only originates from the time splitting and then can be controlled by using high order splittings. 

Denoting $\varphi_t^{\mathcal{H}_\star}({\mathcal Z}(0))$, the exact solution at time $t$ of $\partial_t \mathcal{Z} = \{ \mathcal{Z}, \mathcal{H}_\star \}$ (where $\star=v, E, S_1, S_2, S_3$,) with the initial condition ${\mathcal Z}(t=0)$, 
the solution of the full model \eqref{full_model} 
is thus approximated by 

\begin{equation}
    {\mathcal Z}(t) = \Big( {\Pi_{\star=v, E, S_1, S_2, S_3} } \;\;   
    \varphi^{\mathcal{H}_\star}_t  \Big) \,{\mathcal Z}(0).      
\end{equation}
This is a first-order splitting, but higher order splittings  
could also be derived. Since the splitting involves here 
5 steps, we will restrict ourselves to the Strang scheme 
\begin{equation}
    {\mathcal Z}(t) = \Big(\varphi^{\mathcal{H}_v}_{t/2}\circ \varphi^{\mathcal{H}_E}_{t/2}\circ \varphi^{\mathcal{H}_{S_1}}_{t/2}\circ \varphi^{\mathcal{H}_{S_2}}_{t/2}\circ \varphi^{\mathcal{H}_{S_3}}_{t}\circ \varphi^{\mathcal{H}_{S_2}}_{t/2}\circ \varphi^{\mathcal{H}_{S_1}}_{t/2}\circ \varphi^{\mathcal{H}_E}_{t/2}\circ \varphi^{\mathcal{H}_v}_{t/2} \Big)  \, {\mathcal Z}(0).  
\end{equation}
Such Hamiltonian splitting are known to maintain long term accuracy of the total energy. Moreover, in our case, one can also prove the scheme preserves {\it exactly} the norm of ${\mathbf S}$.

\begin{pro}
The update \eqref{solutionofHS1}, \eqref{solutionofHS2} and \eqref{solutionofHS3} 
of the spin ${\bf S}$ through the hamiltonian splitting discretization preserves the norm of the spin: $||{\bf S}(\cdot, t)||=1$ if $||{\bf S}^0(\cdot)||=1$.
\end{pro}
    \begin{proof}
      By \eqref{solutionofHS1}, \eqref{solutionofHS2} and \eqref{solutionofHS3}, 
        the vector spin ${\bf S}$ is updated through the multiplication 
        of a matrix $\exp(\alpha J t)$ ($J$ being the symplectic matrix) which is a rotation matrix of angle $(-\alpha)$  in $\mathbb{R}^2$. Let introduce the $3\times 3$ matrix $\mathcal{A}$ corresponding to \eqref{solutionofHS3}
        \begin{equation}
            \mathcal{A}=\begin{pmatrix}
                \exp(\alpha_{\mathcal{H}_{S_3}} J  t)& {\bf 0}^T \\
                {\bf 0} & 1
            \end{pmatrix}
        \end{equation}
        with ${\bf 0}=(0,0)$ and $\alpha_{\mathcal{H}_{S_3}}=\frac{\widetilde{K}}{4}\int f_3^0 \mathrm{d}{ v}+A \partial^2_x {S}^{0}_3$. 
        We then reformulate \eqref{solutionofHS3} as 
        ${\bf S}(x, t)=\mathcal{A}{\bf S}^0(x)$ 
        from which we easily deduce the norm is preserved. The same is true for  \eqref{solutionofHS1} and  \eqref{solutionofHS2}. 
        We finally deduce $||S(\cdot, t)||=1$ as long as $||S^0(\cdot)||=1$.
    \end{proof}

\section{Numerical results}
\label{section:results}

In this section, we present some numerical results obtained with the nonlinear code described in section \ref{section:numericalmethod}. The results will also be compared to the analytical linear response, as detailed in section \ref{section:linear}.
In the results presented below, the numerical parameters are chosen as follows (nondimensional units are used everywhere): Number of points in space and velocity $N_x=119, \;  N_v=1024$, time-step $\Delta t=0.1$, variable ranges in the phase space: $v\in [-5, 5],  \; x\in [0, 2\pi/k]$,  perturbation wavenumber $k=0.5$.

The initial condition is a periodic perturbation of the equilibrium $f_0^{(0)}= {\cal F}, \; f_3^{(0)}= \eta {\cal F},\;  f_1^{(0)}=f_2^{(0)}=S_1^{(0)}=S_2^{(0)}=0,\;  S_3^{(0)}=1$, 
where ${\cal F}$ is a spatially homogeneous equilibrium (either a Maxwell-Boltzmann or a two-stream distribution).
This equilibrium represents ions that are fully polarized in the $\ell=3$ direction, while the electrons are partially polarized along the same direction, with a polarization rate equal to $\eta$.

After the perturbation, the initial condition is as follows
\begin{equation}
\begin{aligned}
f_0(t=0^{+}, x, v) &= {\cal F}(v) (1+\varepsilon\cos(kx)), \\
f_1(t=0^{+}, x, v) &= \eta {\cal F}(v) \varepsilon\cos(kx), \\
f_2(t=0^{+}, x, v) &= \eta {\cal F}(v)\varepsilon\sin(kx), \\
f_3(t=0^{+}, x, v) &= \eta {\cal F}(v)(1+\varepsilon\cos(kx)), \\
S_1(t=0^{+}, x) &= \frac{\varepsilon}{\sqrt{1+\varepsilon^2}}\sin(kx)  , \\
S_2(t=0^{+}, x) &= \frac{\varepsilon}{\sqrt{1+\varepsilon^2}}\cos(kx)    , \\
S_3(t=0^{+}, x) &= \frac{1}{\sqrt{1+\varepsilon^2}},  
\end{aligned}
\end{equation}
where the amplitude of the perturbation is $\varepsilon=10^{-3}$. 
Note that the perturbation is chosen such that: $\|\mathbf{S}(t=0, x)\|^2 = S_1^2(0, x)+S_2^2(0, x)+S_3^2(0, x) = 1$.
The nondimensional physical constants are those defined in section \ref{subsec:normalized}, i.e., 
 $A = 0.0148$ (ion-ion magnetic coupling), $\widetilde{K}= 0.161$ (ion-electron magnetic coupling), and $H = 0.339$ (scaled Planck constant). The numerical results will be expressed in terms of the  units defined in section \ref{subsec:normalized}. All logarithms are Neperian (base $e$).

\begin{table}
\centering
\begin{tabular}{|c|c|c|c|}
\hline 
  MB1   & $\eta=\tanh(\widetilde{K}H)\approx 0.0547$ & $\omega_s=0.003680-2.739\times 10^{-5}i$ & $\varepsilon=10^{-3}$ \\
  \hline
   MB2  &  $\eta=0.5$ & $\omega_s=0.02088-0.005253i$ & $\varepsilon=10^{-3}$\\
  \hline
   MB3 & $\eta=-0.5$ & $\omega_s=0.01725+0.006162i$ & $\varepsilon=10^{-6}$\\
  \hline
\end{tabular}
\caption{Main numerical and physical parameters of the three runs that utilize a Maxwell-Boltzmann (MB) equilibrium: initial electron spin polarization $\eta$, ion spin frequency $\omega_s$, and initial perturbation $\varepsilon$. The values of $\omega_s$ are those of the linear response calculation using the ZEAL code. Other values are: $k=0.5$, $N_x=119$, $N_v=1024$, and $v_{\max}=5$.}
\label{tableMB}
\end{table}

\subsection{Maxwell-Boltzmann (MB) equilibrium} \label{subsec:MB}
Here, we consider the Maxwell-Boltzmann equilibrium \eqref{eq:MBequilibrium} that was used for the linear analysis. 
We will analyze three case, for different electron polarizations $\eta$. In the first case (MB1), the polarization is taken to be self-consistent with the ions, i.e., the electron polarization is due solely to the magnetic field generated by the ions, so that $\eta =\tanh(\widetilde{K}H)$ (see Appendix \ref{appendix:equilibrium}).
In the remaining two cases (MB2 and MB3), the polarization will be chosen arbitrarily as $\eta = \pm 0.5$. This polarization may be achieved through the application of an external magnetic field.
The parameters of these Maxwell-Boltzmann simulations are summarized in Table \ref{tableMB}.

\paragraph{MB1.} The roots of the dispersion relation for charges ($\omega_e$) and spins ($\omega_s$), calculated using the ZEAL code, are the following 
\begin{eqnarray}
\label{disp_1f}
\omega_e&=& 1.225-i \, 0.03626 \\
\label{disp_1s}
\omega_s&=& 0.003680- i \, 2.739\times 10^{-5}. 
\end{eqnarray} 
We remark that: (i) the real part of $\omega_e$ is close to the plasma frequency (equal to unity here), while its imaginary part is much smaller, in accordance with the Bohm-Gross dispersion relation; (ii) the real part of  $\omega_s$ is much smaller than the plasma frequency, in accordance with \eqref{eq:scalediffer}, while its imaginary part is even smaller, signifying the almost absence of spin damping.


\begin{figure}
\center{
\subfigure[]{\includegraphics[scale=0.33]{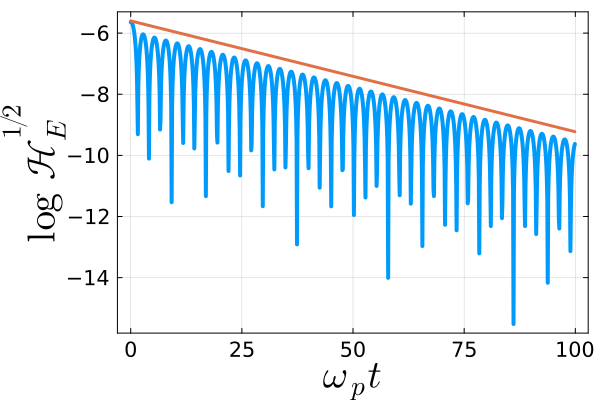}}
\subfigure[]{\includegraphics[scale=0.33]{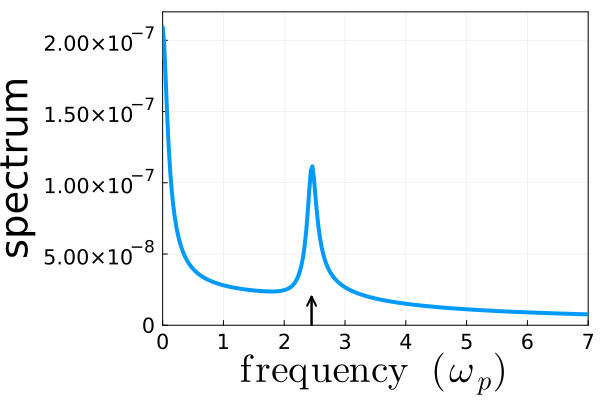}}
\subfigure[]{\includegraphics[scale=0.33]{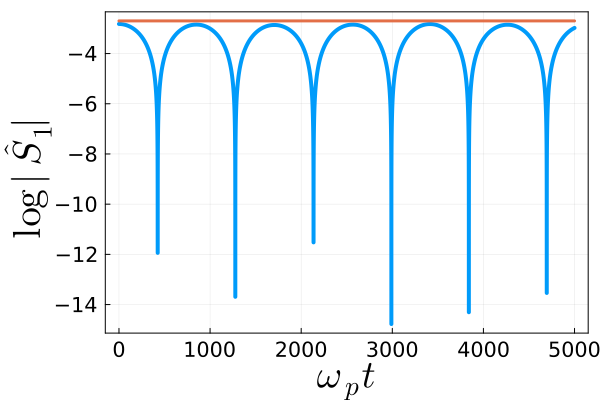}}
\subfigure[]{\includegraphics[scale=0.33]{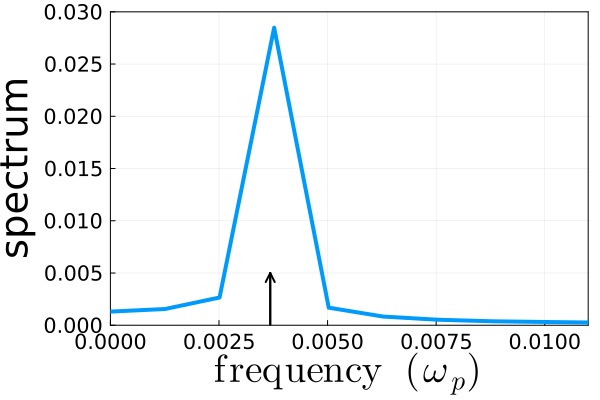}}
}
\caption{MB1 simulation.
Time history of the square root of the electric energy  ${\cal H}_E^{1/2}$ (given by \eqref{eq:splitHamiltonianparts_def}), in semi-$\log$ scale (a)
and corresponding frequency spectrum (b). The red straight line represents the linear damping rate given in \eqref{disp_1f}.
Time history of the absolute value of the real part of the first Fourier mode of the ion spin $\hat{S}_1(k,t)$ in semi-$\log$ scale (c) and corresponding frequency spectrum (d). The red straight line corresponds to zero damping, see \eqref{disp_1s}. The arrows in the spectral plots correspond to the results of linear response theory.}
\label{MB_run1}
\end{figure}

In figures \ref{MB_run1}, we plot the time evolution of some physical quantities associated to the electron charge [panels (a) and (b)] and to the ion spin [panels (c) and (d)]. The Coulomb electric energy decays exponentially with a rate ${\rm Im}\, \omega_e$ very close to the one predicted by the linear response analysis (Landau damping).  The real part of the frequency is also very close to the analytical prediction of \eqref{disp_1f}, with an additional factor of 2 due to the modulus.

In figure \ref{MB_run1} (c),(d),
we show the evolution of the absolute value of the real part of the first Fourier mode of the ion spin $S_1(x,t)$, i.e. $\hat{S}_1(k,t)$, with $k=0.5$ in this case.
In agreement with \eqref{disp_1s}, this mode is virtually undamped (the red line is horizontal and corresponds to zero damping). The corresponding frequency spectrum peaks in the vicinity of the theoretical magnon frequency ${\rm Re}\,  \omega_s$. Note that, due to the great disparity between the magnon and the plasmon frequencies, only a few ($\approx 6$) magnon frequencies could be observed, resulting in a limited accuracy for the magnon spectrum.

In addition to the good agreement with the linear theory for $\omega_e$ and $\omega_s$, we also emphasize that the modulus of the ion spin vector $\|{\bf S}(t, x)\|$ is preserved up to machine accuracy and that the (relative) total energy is preserved up to $10^{-7}$.

\paragraph{MB2.}
For this second test, we consider an initial condition with an electron spin polarization rate $\eta=0.5$. This can be achieved through an external magnetic field $B_3^{\rm ext}$ directed along the same direction as the ion polarization. The positive value of $\eta$ correspond to the "natural" polarization direction for the electrons, parallel to that of the ions and oriented in the same way, as in the self-consistent case. Hence, we expect this equilibrium to be magnetically stable.

As was mentioned earlier, the charge dynamics is decoupled from the spin dynamics in the linear regime, hence the electric response (not shown here) is the same as that of figure \ref{MB_run1}, displaying plasmonic oscillations and Landau damping.

The spin response is depicted in figure \ref{MB_run2}, where we show the first Fourier moments of the ion and electron spins and their frequency spectra. 
In this case, a clear damping of the magnon mode is observed, which is in good agreement with the roots of the dispersion relation: $\omega_s= 0.02088 - i \, 0.005253$, which is to be compared to the damping rate obtained from the simulation, $\gamma = -0.005186$.
The real part of the frequency, see figure \ref{MB_run2}(b), shows a peak near ${\rm Re}\, \omega_s \approx 0.02$, also in good accordance with the linear response result.

The electron spin density ${\mathbf M}$, shown in figure \ref{MB_run2}(c)-(d), follows the same evolution as the ions, with very similar frequency and damping rate.


\begin{figure}
\center{
\subfigure[]{\includegraphics[scale=0.35]{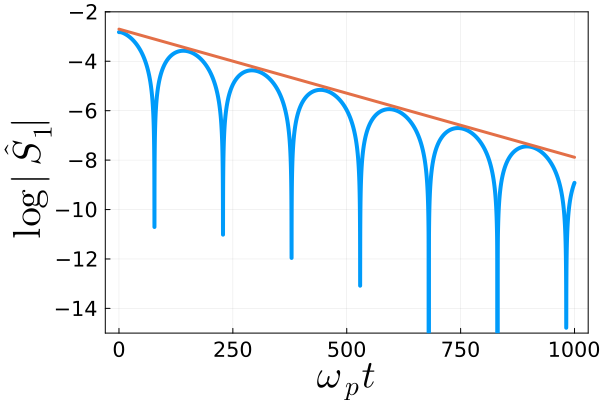}}
\subfigure[]{\includegraphics[scale=0.35]{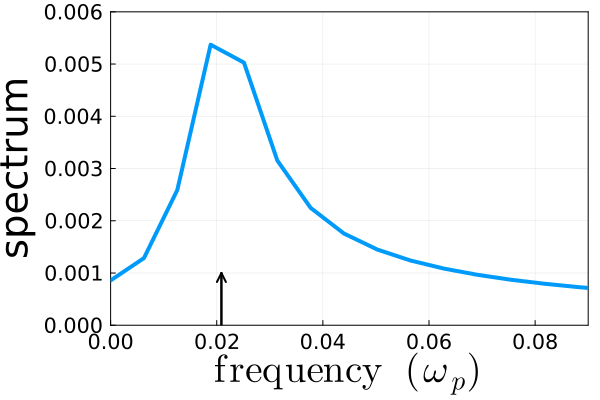}}
\subfigure[]{\includegraphics[scale=0.35]{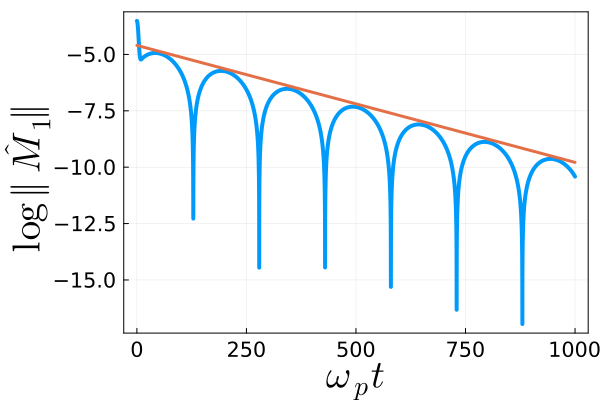}}
\subfigure[]{\includegraphics[scale=0.35]{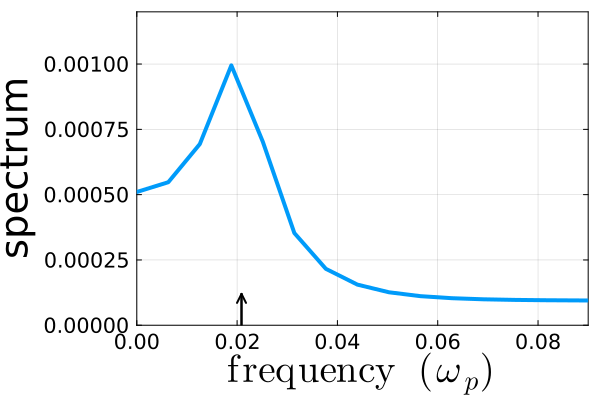}}
}
\caption{MB2 simulation ($\eta=0.5$).
Time history of the absolute value of the real part of the first Fourier mode of the ion spin $\hat{S}_1(k,t)$ in semi-$\log$ scale (a) and corresponding frequency spectrum (b). The slope of the red straight line is $-0.005186$, very close to the linear response result given in Table \ref{tableMB}. The peak of the frequency spectrum also matches the linear result ${\rm Re\, \omega_s} = 0.02088$ (indicated by an arrow on the plot) with good accuracy.    
Panels (c) and (d) show the same quantities for the electronic spin mode $\hat{M}_1(k,t)$. The real and imaginary parts of the frequency are the same as for the ion spins.}
\label{MB_run2}
\end{figure}

\paragraph{MB3.}
Here, we consider an electron gas which is initially polarized in the opposite direction to the one corresponding to the self-consistent case. In this case, the polarization rate is negative, and we take $\eta=-0.5$. Since the electron polarization is opposite to the self-consistent scenario, we expect the system to be unstable, as it attempts to restore the "natural"  direction of polarization.

In figures \ref{MB_run3}(a)-(b), we plot the evolution of the first Fourier mode of the ion spin and its frequency spectrum. The real part of the frequency and the instability rate are very close to the linear response result $\omega_s=0.01725+ i\, 0.006162$. 
After about $2000\, \omega_p^{-1}$, the instability saturates nonlinearly.
The electric field evolution is the same as in figure \ref{MB_run1} (a).

The electron spin density ${\mathbf M}$, shown in figure \ref{MB_run3}(c)-(d) follows the same evolution as the ions, with very similar frequency and instability rate.


\begin{figure}
\center{
\subfigure[]{\includegraphics[scale=0.35]{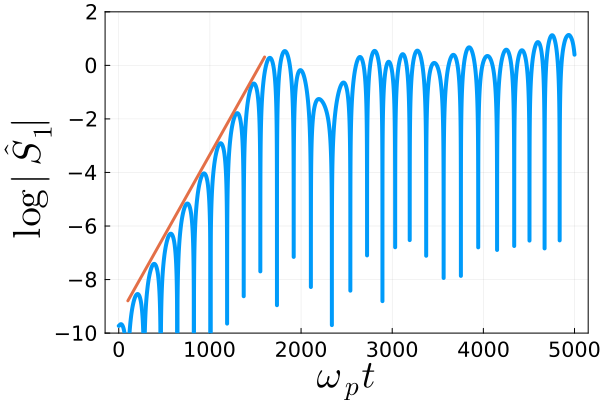}}
\subfigure[]{\includegraphics[scale=0.35]{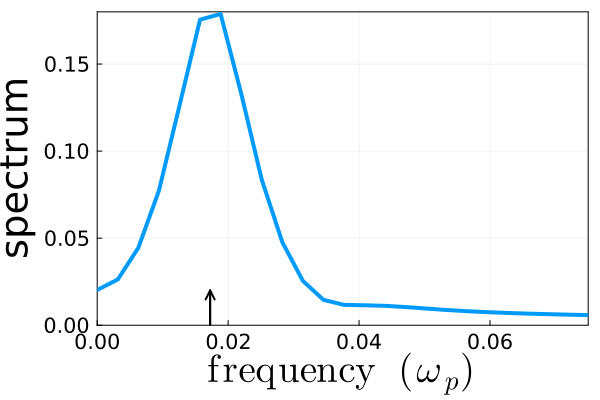}}
\subfigure[]{\includegraphics[scale=0.35]{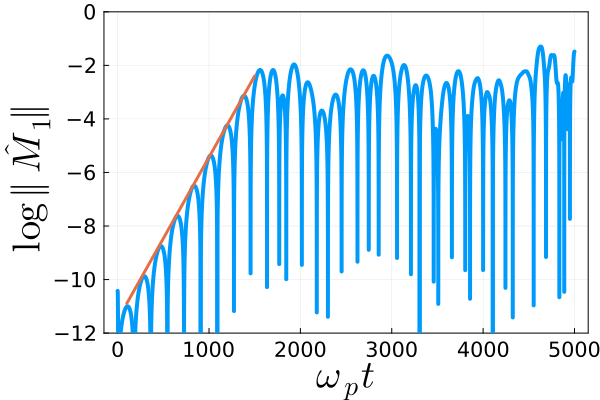}}
\subfigure[]{\includegraphics[scale=0.35]{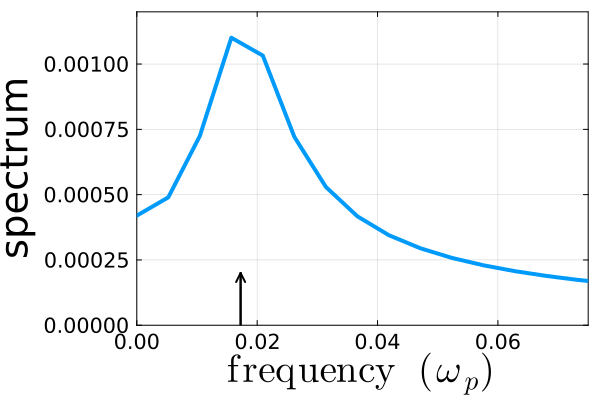}}
}
\caption{MB3 simulation ($\eta=-0.5$).
Time history of the absolute value of the real part of the first Fourier mode of the ion spin $\hat{S}_1(k,t)$ in semi-$\log$ scale (a) and corresponding frequency spectrum (b). The slope of the red straight line is $0.00607$, very close to the linear response result given in Table \ref{tableMB}. The peak of the frequency spectrum also matches the linear result ${\rm Re\, \omega_s} = 0.01725$ (indicated by an arrow on the plot) with good accuracy.
Panels (c) and (d) show the same quantities for the electronic spin mode $\hat{M}_1(k,t)$. The real and imaginary parts of the frequency are the same as for the ion spin.}
\label{MB_run3}
\end{figure}


\subsection{Two-stream (TS) equilibrium} \label{subsec:TS}
In this subsection, we consider a two-stream equilibrium for the initial electron distribution
$${\cal F}(v)=\frac{1}{2\sqrt{\pi}}(e^{-(v-u)^2} + e^{-(v+u)^2}).$$
This equilibrium can be either stable or unstable for the charge dynamics, depending on the value of the stream velocity $u$. In the numerical runs reported below, we have chosen $u=1.4$, which corresponds to a stable case (run TS1), and $u=3$ which corresponds to an unstable case (run TS2 and TS3).  
In TS1 and TS2, we use the self-consistent value for the electron spin polarization, $\eta = \tanh(\widetilde{K}H)\approx 0.0547$, while in TS3 and TS4 we force a spin instability by setting $\eta= - 0.5$.

The parameters of these runs are summarized in Table \ref{tableTS}.

\begin{table}
\centering
\begin{tabular}{|c|c|c|c|}
\hline
  TS1   & $\eta=\tanh(\widetilde{K}H)\approx 0.0547$ & $u=1.4$ & $N_v=1024, v_{\max}=8, \varepsilon=10^{-3}$\\
    \hline
   TS2  &  $\eta=\tanh(\widetilde{K}H)\approx 0.0547$ & $u=3$ & $N_v=1536, v_{\max}=12, \varepsilon=10^{-6}$\\
     \hline
   TS3  &  $\eta=-0.5$ & $u=3$ & $N_v=512, v_{\max}=14, \varepsilon=10^{-6}$\\
    \hline
   TS4  &  $\eta=-0.5$ & $u=1.4$ & $N_v=512, v_{\max}=14, \varepsilon=10^{-6}$\\
    \hline
\end{tabular}
\caption{Main numerical and physical parameters of the runs that utilize a two-stream equilibrium. Other values are: $k=0.2$ and $N_x=129$.}
\label{tableTS}
\end{table}

\paragraph{TS1.} In this case, the stream velocity is weak ($u=1.4$) so that the charge sector of the dynamics is basically undamped, as seen on figure  \ref{TSI_run4}(a) for the electric field. 
The spin sector is more interesting, both for the ions and the electrons, which are rather strongly damped at a rate $\approx 6.3 \times 10^{-4}$. 
This is in contrast with the corresponding Maxwell-Boltzmann simulation (MB1, figure \ref{MB_run1}) where the spin mode was very weakly damped. Although the wavenumber is not the same ($k=0.5$ for MB1 and $k=0.2$ for TS1), it appears that the equilibrium profile has a strong impact on the stability properties of the ion magnon mode.


\begin{figure}
\center{
\subfigure[]{\includegraphics[scale=0.35]{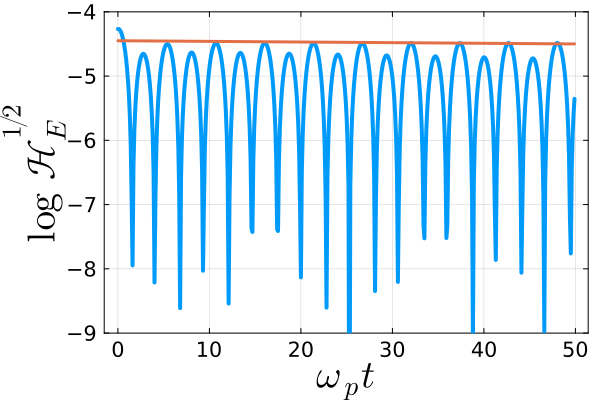}}
\subfigure[]{\includegraphics[scale=0.35]{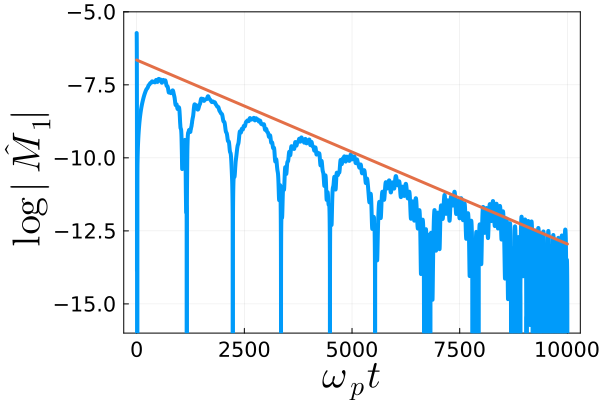}}
\subfigure[]{\includegraphics[scale=0.35]{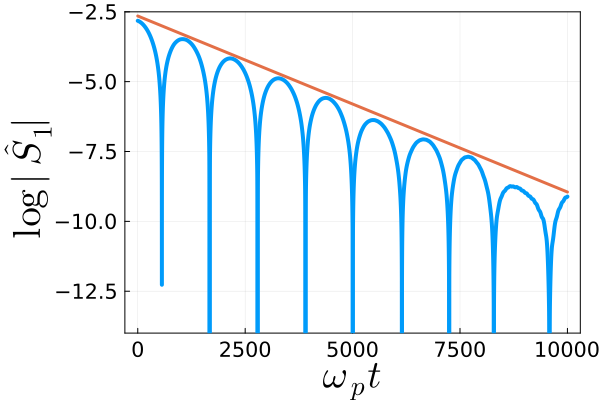}}
}
\caption{TS1 simulation. (a): time evolution of the square root of the electric energy ${\cal H}_E^{1/2}$ (given by \eqref{eq:splitHamiltonianparts_def}) for short times $t\in [0, 50]$. (b): time evolution of the absolute value of the real part of the fundamental mode of the electron spin  $\hat{M}_1$. 
    (c): time evolution of the absolute value of the real part of the fundamental mode of the ion spin  $\hat{S}_1$. 
    The red straight lines have slopes equal to zero for the electric energy and $-6.3 \times 10^{-4}$ for  $\hat{M}_1$ and $\hat{S}_1$. }
\label{TSI_run4}
\end{figure}

\begin{figure}
\center{
\subfigure[]{\includegraphics[scale=0.35]{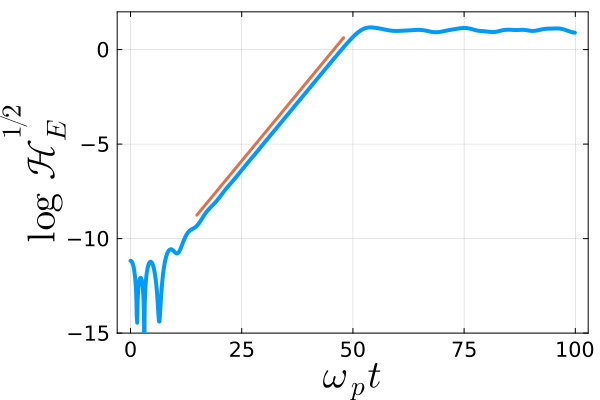}}
\subfigure[]{\includegraphics[scale=0.35]{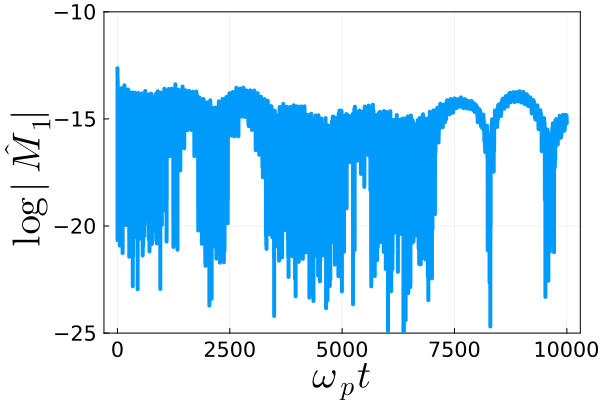}}
\subfigure[]{\includegraphics[scale=0.35]{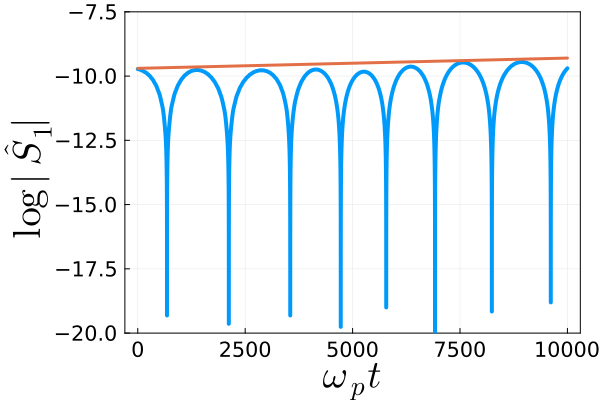}}
}
\caption{TS2 simulation. (a): time evolution of the square root of the electric energy ${\cal H}_E^{1/2}$ (given by \eqref{eq:splitHamiltonianparts_def}) for short times $t\in [0, 100]$. (b): time evolution of the absolute value of the real part of the fundamental mode of the electron spin  $\hat{M}_1$. 
(c): time evolution of the absolute value of the real part of the fundamental mode of the ion spin  $\hat{S}_1$. 
The red straight lines have slopes equal to $0.2845$ for the electric energy and $4\times 10^{-5}$ for  $\hat{S}_1$. }
\label{TSI_run5}
\end{figure}

\begin{figure}
\center{
\subfigure[]{\includegraphics[scale=0.35]{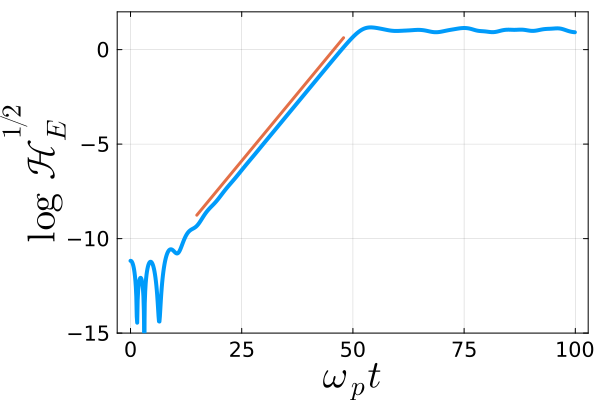}}
\subfigure[]{\includegraphics[scale=0.35]{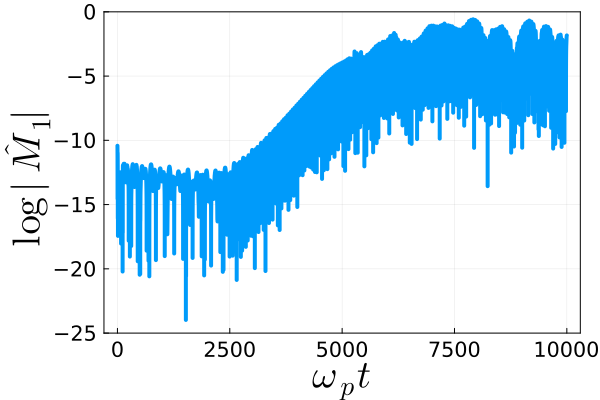}}
\subfigure[]{\includegraphics[scale=0.35]{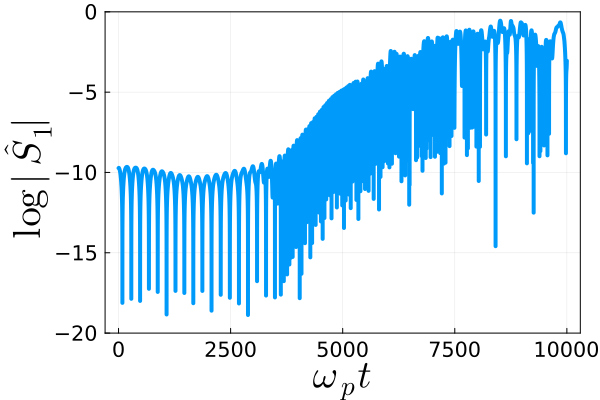}}
}
\caption{TS3 simulation. (a): time evolution of the square root of the electric energy ${\cal H}_E^{1/2}$ (given by \eqref{eq:splitHamiltonianparts_def}) for short times $t\in [0, 100]$; the red straight line has slopes equal to $0.2845$. (b): time evolution of the absolute value of the real part of the fundamental mode of the electron spin  $\hat{M}_1$. 
(c): time evolution of the absolute value of the real part of the fundamental mode of the ion spin  $\hat{S}_1$.  }
\label{TSI_run11}
\end{figure}

\begin{figure}
\center{
\subfigure[]{\includegraphics[scale=0.5]{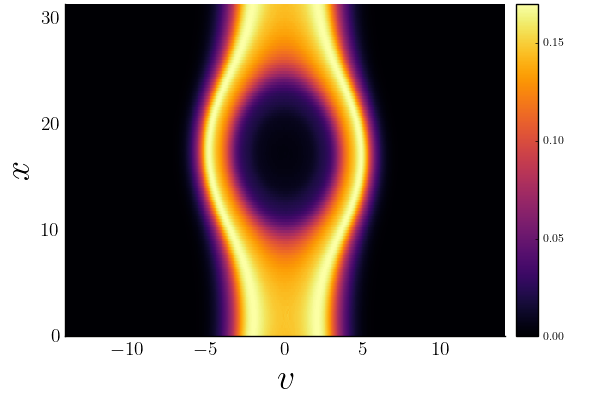}}
\subfigure[]{\includegraphics[scale=0.5]{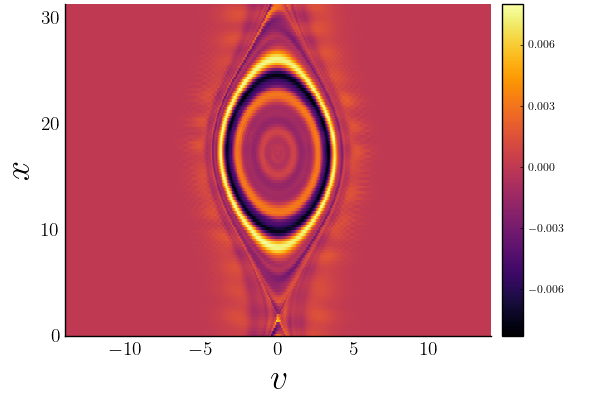}}
\subfigure[]{\includegraphics[scale=0.5]{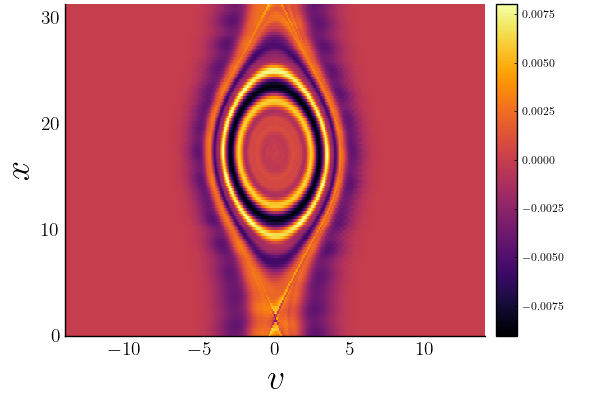}}
\subfigure[]{\includegraphics[scale=0.5]{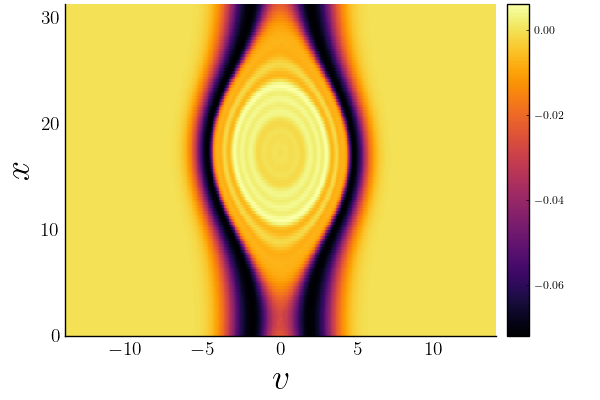}}
}
\caption{TS3 simulation. Contour plots of the distribution functions in the $(x,v)$ phase space at the final time $\omega_p t=10^4$: $f_0$ (a), $f_1$ (b), $f_2$ (c), and  $f_3$ (d).   }
\label{fig:phasespace}
\end{figure}

\begin{figure}
\center{
\subfigure[]{\includegraphics[scale=0.35]{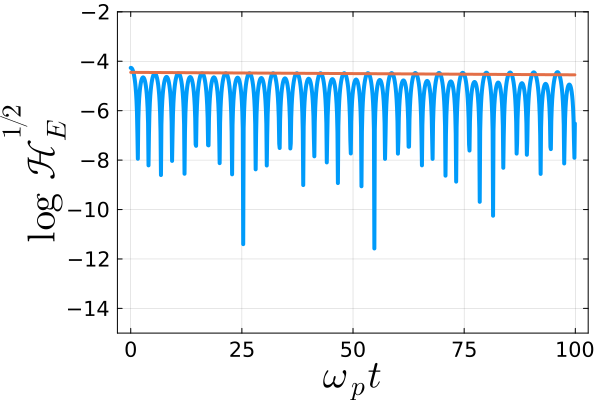}}
\subfigure[]{\includegraphics[scale=0.35]{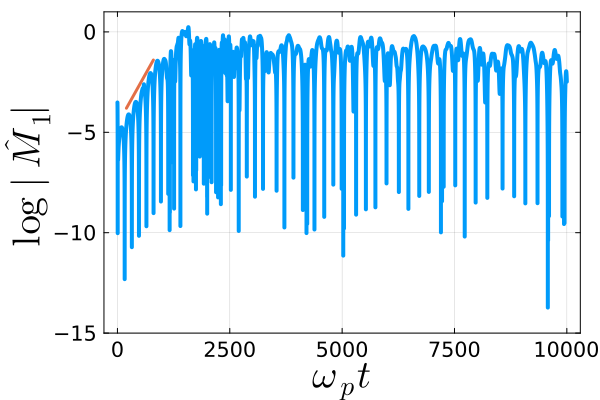}}
\subfigure[]{\includegraphics[scale=0.35]{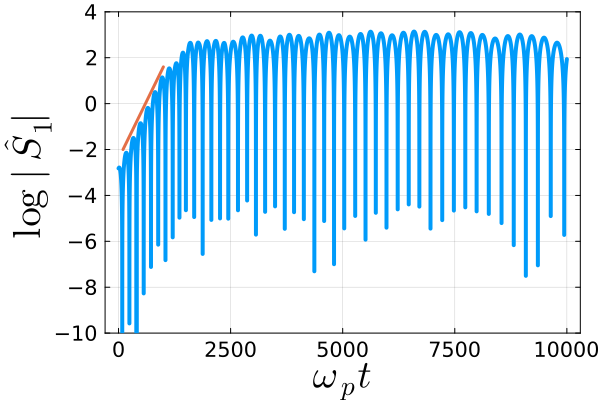}}
}
\caption{TS4 simulation. (a): time evolution of the square root of the electric energy ${\cal H}_E^{1/2}$ (given by \eqref{eq:splitHamiltonianparts_def}) for short times $t\in [0, 100]$. (b): time evolution of the absolute value of the real part of the fundamental mode of the electron spin  $\hat{M}_1$. 
(c): time evolution of the absolute value of the real part of the fundamental mode of the ion spin  $\hat{S}_1$.  The red straight lines have slope equal to $0.004$.}
\label{TSI_run12}
\end{figure}

\paragraph{TS2.}
This run uses the same parameters as TS1, except that the stream velocity is larger, $u=3$. We also changed the magnitude of the initial perturbation, now set to $\varepsilon=10^{-6}$, in order to get a longer-lasting linear phase. Linear theory predicts an instability in the charge sector, with growth rate equal to $0.2845$, which is confirmed by the numerical data shown in figure \ref{TSI_run5}(a).
The ion spin sector displays a very weak instability, with an observed growth rate $\approx 4\times 10^{-5}$.  The electron spin remains at very low amplitude all along the simulation time.

\paragraph{TS3.}
Here, we wish to consider a case where an instability is expected {\em both} in the charge and in the spin sectors. Therefore, we take the same value $u=3$ for the stream velocity as in TS2, and an electron spin polarization $\eta=-0.5$, which led to the instability of the magnon mode in MB3.
The results are plotted in figure \ref{TSI_run11} and show that the evolution of the electric field is almost the same as in the case TS2. This is natural, as the linear response of the charge sector is independent of $\eta$ (nonetheless, one may have expected some differences after the nonlinear regime is attained, around $\omega_p t= 50$, but in practice the two curves are very similar, although not identical).
Interestingly, the electron $\hat{M}_1$ and ion $\hat{S}_1$ spins are initially stable until $\omega_p t \approx 2500$, i.e., well into the nonlinear regime, and only become unstable later. Their growth rate is much smaller than the one associated with the charges.

The phase space portraits at the end of the simulation are displayed in figure \ref{fig:phasespace}, for the four distributions $f_0(t, x, v)$ and $f_\ell(t, x, v),\, \ell=1, 2, 3$. Typically for this type of instability, the two-stream structure has been destroyed in the nonlinear regime and a single vortex centered at $v=0$ can be observed. The vortex is present not only in the charge distribution $f_0$, but also in the spin distributions ${\mathbf f}$.

\paragraph{TS4.}
Finally, we repeat the same simulation  as TS3, but for a smaller stream velocity $u=1.4$, so that there is no instability in the charge sector (see figure \ref{TSI_run12}). In this case, the usual magnon instability ($\eta<0$) develops immediately, in contrast to the preceding TS3 case. Although it is difficult to draw definite conclusions, it is clear that the onset, or otherwise, of a charge instability interacts strongly with the development of a magnon instability. This is  further evidence that the charge and spin sectors are closely intertwined and need to be both included in the model for an accurate description of the magnonic dynamics.

\section{Conclusion} \label{section:conclusion}
In this work, we have built on previous developments 
\cite{Crouseilles2021,Crouseilles2023,Manfredi2023} to construct a fully kinetic 1D model of the interaction between the charge and the spin dynamics in a material with intrinsic magnetization (ferromagnet). The electron dynamics is described by a four-component phase space distribution function $f_0(t, x, v)$, $f_\ell(t, x, v), \ell=1, 2, 3$, where $f_0$ is related to the electron charge and $f_\ell$ to the electron spin polarization in the $\ell$ direction. The fixed ions are modeled by the Landau-Lifshitz equation for the magnetization ${\mathbf S}(t,x)$. The electron charges interact through the self-consistent electric field, solution of the Poisson equation. The electron and ion spins interact through the magnetic exchange, whose magnitude is controlled by the coupling constant $K$. Finally, the ion spins interact among themselves via the ion-ion magnetic exchange, with coupling constant $J$.

This model can be seen as an extension of the standard Vlasov-Poisson equations for mobile electrons and fixed ions, taking into account the electron spin and allowing for a spin dynamics for the ions.

We first focused on the linear response of this system when the equilibrium is a Maxwell-Boltzmann function. The full dispersion relation is rather complex, but can be split into a charge sector and a spin sector. The former is independent of the spin and leads to the standard Bohm-Gross relation. 
The spin sector was analyzed more in detail, particularly the occurrence of damping and instability when the ion-electron magnetic coupling constant and the electron spin polarization at equilibrium are varied. 
Interestingly, we observe damping when the electron spin polarization is directed along their "natural" direction of magnetization (the one dictated by the magnetic field generated by the ions) and instability when it is directed opposite to it.

Next, we built a computational code based on the Hamiltonian splitting method first developed in Ref. \cite{Crouseilles2023, laser2020}. This is an Eulerian grid-based method that solves simultaneously the coupled Vlasov-Poisson-Landau-Lifshitz equations. This technique allowed us to achieve great accuracy for the conserved quantities: the modulus of the ion spin vector $\|{\bf S}(t, x)\|$ is preserved up to machine accuracy and the (relative) total energy is preserved up to $\approx 10^{-7}$.  

We have used the code to validate the estimations of the linear response theory, with very good agreement between the two approaches for Maxwell-Boltzmann equilibria. We also tested it on two-stream equilibria, which may lead to instability in the charge sector, depending on the streams' relative velocities. Particularly interesting was the case where an instability in the charge sector leads to a much delayed instability in the spin sector, which develops well after the charge dynamics has saturated nonlinearly. This is further evidence of the close interaction between the charge and spin sectors in the coupled plasmon-magnon dynamics.

The Maxwell-Boltzmann equilibria and parameter range used in this work, with densities close to those of solids ($\approx 10^{29} \rm m^{-3}$) and temperatures of the order of 10~eV, are relevant to the warm dense matter (WDM) regime \cite{Bonitz2020} that appears, among others, in inertial fusion experiments.
For these conditions, the electron plasma is weakly degenerate ($T_e \approx T_F$), so that it can be characterized with relatively good accuracy by a MB distribution. The ions are fixed and non-degenerate. In this WDM regime, ultrafast nonequilibrium dynamics has been recently observed  thanks to subpicosecond laser pulses \cite{Falk2018}. At these very short timescales, and for magnetic materials, the electron and ion spin polarization may not yet be lost, and impact the early instants of the dynamics.

However, MB distributions are not relevant to condensed-matter systems -- for which the Fermi temperature is well above the room temperature -- and the latter should therefore be described by a Fermi-Dirac (FD) equilibrium. Calculations of the dispersion relation for FD distributions are notoriously more involved than for MB distributions, particularly in the finite-temperature case. These developments are left for future work.

\section*{Acknowledgments}
This work is supported by France 2030 government investment plan managed by the French National Research Agency under grant reference PEPR SPIN -- [SPINTHEORY] ANR-22-EXSP-0009.
This work was partially funded by the French National Research Agency (ANR) through the Programme d'Investissement d'Avenir under contract ANR-11-LABX-0058-NIE and ANR-17-EURE-0024 within the Investissement d'Avenir program ANR-10-IDEX-0002-02.

\newpage

\appendix

\section{Appendix: Spin-polarized equilibrium}
\label{appendix:equilibrium}

To compute stationary states, it is more convenient to go back to the standard representation of the Wigner function \cite{Manfredi2019}:
\begin{equation}
\mathcal{F} = 
\begin{pmatrix}
    \mathcal{F}_{++}& \mathcal{F_{+-}} \\
     \mathcal{F_{-+}} & \mathcal{F_{--}} 
\end{pmatrix} ,
\end{equation}
where $+$ ($-$) stands for spin-up (spin-down) with respect to the direction $\ell=3$. The relationship between this representation and the Pauli representation used in the main text is the following: $f_\ell = {\rm tr} (\sigma_\ell \mathcal{F}), \,  f_0 = {\rm tr} (\mathcal{F})$, 
where $\sigma_\ell\, (\ell=1,2,3)$ are the Pauli matrices:
$$
\sigma_1 = 
\begin{pmatrix}
    0 & 1 \\
     1 & 0
\end{pmatrix}, \quad
\sigma_2 = 
\begin{pmatrix}
    0 & -i \\
     i & 0
\end{pmatrix}, \quad
\sigma_3 = 
\begin{pmatrix}
    1 & 0 \\
     0 & -1
\end{pmatrix}.
$$

For a spatially homogeneous equilibrium, the terms corresponding to the self-consistent electric energy $\mathcal{H}_{E}$ and the spin energy $\mathcal{H}_{spin}$ vanish from the expression of the Hamiltonian \eqref{hamiltonian_dim}. 
In the above basis, the Hamiltonian is a diagonal $2 \times 2$ matrix ${\rm diag}( \mathcal{H}_+, \mathcal{H}_{-}) $, 
where ${\cal H}_\pm = \frac{m}{2} v^2 \pm \mu_B B_3$ is the signed sum of the kinetic and Zeeman energies and $B_3$ is the magnetic field generated by the (fully polarized ions), see \eqref{eq:magn-field-ions}. In our dimensionless units $B_3 = - \widetilde{K}/2$, and we get for the Hamiltonian: ${\cal H}_\pm = v^2\mp H \widetilde{K}$.

For a stationary state, the distribution function must be a function of the Hamiltonian, i.e., in the Maxwell-Boltzmann case, $\mathcal{F} = C\,\exp(-\beta \mathcal{H})$, where $C$ is a normalization constant. Hence, the distribution function is also diagonal, with $\mathcal{F}_{++} = C\,\exp(-\beta \mathcal{H}_{+})$, and similarly for $\mathcal{F}_{--}$, where $\beta = 1/(k_B T_e) = 1$ in our units.

Going back to the Pauli basis utilized in the main text, we obtain
\begin{eqnarray*}
f_0(v) &=& {\cal F}_{++} + {\cal F}_{--} = 2C\, e^{-v^2} \cosh(H\widetilde{K}),  \nonumber\\
f_3(v) &=& {\cal F}_{++} - {\cal F}_{--} = 2C\, e^{-v^2} \sinh(H\widetilde{K}) .
\end{eqnarray*}
With the normalization $\int f_0(v)dv=1$, we get $C=1/(2\sqrt{\pi}\cosh(H\widetilde{K}))$. 
As a consequence, the equilibrium distribution function becomes:
$$
f_0(v) = \frac{e^{-v^2}}{\sqrt{\pi}}, \quad f_3(v) = \eta \,\frac{e^{-v^2}}{\sqrt{\pi}} ,
$$
with $\eta=\tanh(H\widetilde{K})$, which is identical to equation \eqref{eq:MBequilibrium} in the main text.

Finally, if the magnetic field in the Hamiltonians ${\cal H}_\pm$ is not the one generated self-consistently by the ions, but instead an external one $B_3^{\rm ext}$, then the electron spin polarization is $\eta=\tanh(2\mu_B B_3^{\rm ext}/ k_B T_e)$ and can take any values in $[-1,1]$.
Note that $\eta>0$ corresponds to a case where the ion spin $\mathbf{S}(t,x)$ and electrons spin $\mathbf{M}(t,x) =  \frac{\hbar}{2} \int \mathbf{f}(t, x, v) dv$ are aligned along the same direction, which is a stable ferromagnetic equilibrium. In contrast, when $\eta<0$, the ion and electron spins point into opposite directions, leading to an unstable equilibrium. This is confirmed by the simulations reported in section \ref{subsec:MB}.

\section{Appendix: Dispersion relation details}\label{appendix:dispersion}
In this Appendix, some details are given about the analytical dispersion relation. In particular, a new form of the dispersion function $D(\omega,k)$ is presented and its derivatives are computed explicitly.

\subsection{Alternative form of the dispersion relation}
\label{appendix:compact_dispersion_relation}
The dispersion relation \eqref{dispersion} writes as:
\begin{align*}
D_{{S}}(\omega,k)&=& -\left[\omega + \frac{c_0}{k} \left[Z\left(\frac{\omega+\widetilde{K}/2}{k}\right)+Z\left(\frac{\omega-\widetilde{K}/2}{k}\right)\right]-c_1 \left[Z'\left(\frac{\omega-\widetilde{K}/2}{k}\right)-Z'\left(\frac{\omega+\widetilde{K}/2}{k}\right)\right]\right]^2\\
&&\hspace{-1.2cm}+\left[A k^2+d +\frac{c_0}{k}  \left[Z\left(\frac{\omega+\widetilde{K}/2}{k}\right)-Z\left(\frac{\omega-\widetilde{K}/2}{k}\right)\right]+ c_1 \left[Z'\left(\frac{\omega-\widetilde{K}/2}{k}\right)+Z'\left(\frac{\omega+\widetilde{K}/2}{k}\right)\right]\right]^2
\end{align*}with $c_0=\widetilde{K}^2\eta /16$,  $c_1=\widetilde{K}^2 H /16$, $d=\widetilde{K}\eta/4 $ (recall that $\eta=\int f_3^{(0)} \mathrm{d}v$). Factorizing leads to
\begin{align*}
    D_{{S}}(\omega,k)=-&\left[\omega+Ak^2+d+\frac{2c_0}{k}Z\left(\frac{\omega+\widetilde{K}/2}{k}\right)+2c_1 Z'\left(\frac{\omega+\widetilde{K}/2}{k}\right) \right]
    \\ \times
    &\left[\omega-Ak^2-d+\frac{2c_0}{k}Z\left(\frac{\omega-\widetilde{K}/2}{k}\right)-2c_1 Z'\left(\frac{\omega-\widetilde{K}/2}{k}\right)\right] .
\end{align*}
Naming $D_{+}$ the first term on the right-hand side and $D_{-}$ the second term, $D_-(-\omega^{*},k)$ can be computed, where the asterisk denotes the complex conjugate: 
\begin{align*}
    D_-(-\omega^*,k)&=-\omega^*-Ak^2-d+\frac{2c_0}{k}Z\left(\frac{-\omega^*-\widetilde{K}/2}{k}\right)-2c_1 Z'\left(\frac{-\omega^*-\widetilde{K}/2}{k}\right)
    \\
    &=-\omega^*-Ak^2 -d+\frac{2c_0}{k}Z\left(-\left(\frac{\omega+\widetilde{K}/2}{k}\right)^*\right)-2c_1 Z'\left(-\left(\frac{\omega+\widetilde{K}/2}{k}\right)^{*}\right) .
\end{align*}
Now, some symmetries in $Z(-z^*)$ and $Z'(-z^*)$  can be used \cite{Fried1961}: $Z(-z^*)=-Z^*(z)$ so $Z'(-z^*) = -2\left(1-z^* Z(-z^*)\right)=-2\left(1+z^* Z^*(z)\right)=Z'^*(z)$.

$D_-(-\omega^*,k)$ is finally expressed as:
\begin{align*}
    D_-(-\omega^*,k)&=-\left[\omega+Ak^2+d+\frac{2c_0}{k}Z\left(\frac{\omega+\widetilde{K}/2}{k}\right)+2c_1Z'\left(\frac{\omega+\widetilde{K}/2}{k}\right)\right]^*  =-D_{+}^{*}(\omega,k) .
\end{align*}
Then, we get: $D_S(\omega,k)=D_-(-\omega^*,k)D_-(\omega,k)$. Hence, if $\omega$ satisfies $D_S(\omega,k)=0$ then $D_S(-\omega^*,k)$ also vanishes. Therefore, we will consider 
\begin{equation*}
    D_{-}(\omega,k)=\omega-Ak^2-d+\frac{2c_0}{k}Z\left(\frac{\omega-\widetilde{K}/2}{k}\right)-2c_1 Z'\left(\frac{\omega-\widetilde{K}/2}{k}\right) 
\end{equation*}
as the dispersion relation  instead of $D_S$. Since $Z'(z)=-2(1+zZ(z))$, 
\begin{equation*}
    D_{-}(\omega,k)=\omega-Ak^2-d+4c_1+Z(z)\left[\frac{2c_0}{k}+4c_1 z \right], 
\end{equation*}
with $z=(\omega-\widetilde{K}/2)/k$. 
Using the expressions of $d,c_0,c_1$ in terms of $\widetilde{K}$ we obtain
\begin{equation}
\label{d-_app}
    D_{-}(\omega,k)=\omega-Ak^2-\frac{\widetilde{K}\eta}{4}+\frac{\widetilde{K}^2H}{4}+Z(z)\left[\frac{\widetilde{K}^2\eta}{8k}+\frac{\widetilde{K}^2H}{4}z\right], 
\end{equation}
which can also be interpreted as a function of $(\omega, \widetilde{K})$ for a constant value of $k$. 

\subsection{Computation of the derivatives of $D_{-}$}
\label{appendix:D_differential}
The partial derivatives of $D_{-}(\omega,\widetilde{K})$ 
given by \eqref{d-_app} with respect to $\omega$ and $\widetilde{K}$ can be computed as follows (with $\eta=\tanh(H\widetilde{K}$) and $z =\frac{\omega-\widetilde{K}/2}{k}$) :
\begin{align*}
    \derv{D_-}{\widetilde{K}} &=-\frac{\eta}{4}-\frac{\widetilde{K}H(1-\eta^2)}{4}+\frac{\widetilde{K}H}{2}
    +\frac{\widetilde{K}^2\eta}{8k^2}+\frac{\widetilde{K}^2Hz}{4k}
    \\
    &+Z(z)\left[
    \frac{\widetilde{K}^2\eta z}{8k^2}+\frac{\widetilde{K}^2Hz^2}{4k}
    +\frac{\widetilde{K}\eta}{4k}+\frac{\widetilde{K}^2H(1-\eta^2)}{8k}+\frac{\widetilde{K}Hz}{2}-\frac{\widetilde{K}^2H}{8k}
    \right]
    \\
    \derv{D_-}{\omega} &=1-\frac{\widetilde{K}^2\eta}{4k^2}-\frac{\widetilde{K}^2Hz}{2k}+Z(z)\left[-\frac{\widetilde{K}^2\eta z}{4k^2}-\frac{\widetilde{K}^2Hz^2}{2k}+\frac{\widetilde{K}^2H}{4k}    \right].
\end{align*}

\section{Appendix: Time splitting}
\label{split_app}
In this Appendix, we give the details of the time solution of the different subsystems induced by the Hamiltonian splitting, as detailed in section \ref{section:numericalmethod}. 
{Regarding the space approximation, Fourier spectral methods are used, so that the linear transport operators (for the Vlasov part) and the elliptic operators (for the Poisson equation) reduce to a simple multiplication in the Fourier space. In the velocity direction, the linear transport operators in the Vlasov equations  are approximated by using a semi-Lagrangian method based on finite volumes (see \cite{crouseilles2010conservative} 
for more details). Finally, all the integrals in velocity space are approximated by standard 
rectangle quadratures. }

\subsection{Subsystem for $\mathcal{H}_v$}

The subsystem $\frac{\partial \mathcal{Z}}{\partial t} = \{ \mathcal{Z}, \mathcal{H}_v \}$ associated to $\mathcal{H}_{v} = \frac{1}{2}\int v^2 f_0 \mathrm{d}{ x}\mathrm{d}{v}$ is

\begin{equation}\label{eq:Hv}
\left\{
\begin{aligned}
&\frac{\partial f_0}{\partial t} = \{f_0, \mathcal{H}_{v} \} = -v\frac{\partial f_0}{\partial x} \\
&\frac{\partial f_\ell}{\partial t} = \{f_\ell, \mathcal{H}_{v} \}= -v\frac{\partial f_j}{\partial x},\ \ell=1,2,3 \\
&\frac{\partial {\mathbf S}^{}}{\partial t}  = \{ {\mathbf S}^{}, \mathcal{H}_{v} \} = {\mathbf 0}\\
&\partial_x^2 V_H=\int f_0 \mathrm{d}{ v}-1.
\end{aligned}
\right.
\end{equation}
We denote the initial value as $ (f_0^0(x,v),\mathbf{f}^0(x,v), {\mathbf S}^{0}(x))$ at time $t=0$. The solution at time $t$ of this subsystem can be written explicitly: 
\begin{equation}
	\begin{aligned}
		&f_0(t,x,v)=f_0^0(x-vt,v), \;\;  \mathbf{f}(t,x,v)=\mathbf{f}^0(x-vt,v),  \;\; {\mathbf S}^{}(t,x)={\mathbf S}^{0}(x).  
	\end{aligned}
\end{equation}

\subsection{Subsystem for $\mathcal{H}_E$}

The subsystem $\frac{\partial \mathcal{Z}}{\partial t} = \{ \mathcal{Z}, \mathcal{H}_E \}$ associated to $\mathcal{H}_E=\frac{1}{2} \int \Big(\frac{\partial V_H}{\partial x}\Big)^2 \mathrm{d}{x}$ is

\begin{equation}\label{eq:HE}
\left\{
\begin{aligned}
&\frac{\partial f_0}{\partial t} = \{f_0, \mathcal{H}_{E} \} = -\frac{\partial V_H}{\partial x} \frac{\partial f_0}{\partial v}\\
&\frac{\partial f_\ell}{\partial t} = \{f_\ell, \mathcal{H}_{E} \}= -\frac{\partial V_H}{\partial x} \frac{\partial f_\ell}{\partial v},\ \ell=1,2,3 \\
&\frac{\partial {\mathbf S}^{}}{\partial t}  = \{ {\mathbf S}^{}, \mathcal{H}_{E} \} = {\mathbf 0}.
\end{aligned}
\right.
\end{equation}
With the initial value $(f_0^0(x,v),\mathbf{f}^0(x,v), {\mathbf S}^{0}(x))$ at time $t=0$, 
the solution at time $t$ is as follows
\begin{equation}
	\begin{aligned}
		&f_0(t,x,v)=f_0^0 \left( x,v-t\frac{\partial V_H}{\partial x}(x) \right), \; {\mathbf f}(t,x,v)={\mathbf f}^0 \left( x,v-t\frac{\partial V_H}{\partial x}(x) \right),  \; {\mathbf S}^{}(t,x)={\mathbf S}^{0}(x). 
	\end{aligned}
\end{equation}

\subsection{Subsystem for $\mathcal{H}_{S_1}$}

The subsystem $\frac{\partial \mathcal{Z}}{\partial t} = \{ \mathcal{Z}, \mathcal{H}_{S_1} \}$ associated to $\mathcal{H}_{S_1}= H \int  f_1 B_1  \mathrm{d}x\mathrm{d}v+ AH \int \left(\frac{\partial {S^{}_1}}{\partial x}\right)^2 \mathrm{d}{ x}$ is
\begin{equation}\label{eq:HBS1}
\left\{
\begin{aligned}
&\frac{\partial f_0}{\partial t} = \{f_0, \mathcal{H}_{S_1} \} = H\frac{\partial B_1}{\partial x} \frac{\partial f_1}{\partial v}\\
&\frac{\partial f_1}{\partial t} = \{f_1, \mathcal{H}_{S_1} \} = H\frac{\partial B_1}{\partial x} \frac{\partial f_0}{\partial v} \\
&\frac{\partial f_2}{\partial t} = \{f_2, \mathcal{H}_{S_1} \} = -B_1 f_3 \\
&\frac{\partial f_3}{\partial t} = \{f_3, \mathcal{H}_{S_1} \} = B_1 f_2\\
&\frac{\partial {S}^{}_1}{\partial t}  = \{ {S}_1, \mathcal{H}_{S_1} \} = {0}\\
&\frac{\partial {S}^{}_2}{\partial t}  = \{ {S}_2, \mathcal{H}_{S_1} \} = \frac{\widetilde{K}}{4}S^{}_3 \int f_1 \mathrm{d}{ v} +A {S}^{}_3 \partial^2_x {S}_1\\
&\frac{\partial {S}^{}_3}{\partial t}  = \{ {S}_3, \mathcal{H}_{S_1} \} = -\frac{\widetilde{K}}{4}S^{}_2 \int f_1 \mathrm{d}{ v} -A {S}^{}_2 \partial^2_x {S}^{}_1.
\end{aligned}
\right.
\end{equation}
with the initial value 
$(f_0^0(x,v),{\mathbf f}^0(x,v),{\mathbf S}^{0}(x))$ at time $t=0$ and ${\mathbf B}^0=-\frac{\widetilde{K} {\mathbf S^{0}}}{2}$. By using ${S}^{}_1={S}^{0}_1$, $B_1=B_1^0$ and $\int f_1 \mathrm{d}{ v}=\int f_1^0 \mathrm{d}{ v}$, we reformulate the equations \eqref{eq:HBS1} as
\begin{eqnarray}
\label{HS1_a}
\displaystyle \partial_t	\begin{pmatrix}
	f_0 \\
	f_1
\end{pmatrix}-H \frac{\partial B_1^0}{\partial x} \begin{pmatrix}
	0 & 1 \\
	1 & 0
\end{pmatrix} \partial_v	\begin{pmatrix}
f_0 \\
f_1
\end{pmatrix} &=&0, \\
\label{HS1_b}
\partial_t	\begin{pmatrix}
			f_2  \\
			f_3
		\end{pmatrix}+B_1^0 J \begin{pmatrix}
f_2  \\
f_3
 \end{pmatrix} &=&0,
\\
\label{HS1_c}
\partial_t	\begin{pmatrix}
			{S}^{}_2  \\
			{S}^{}_3
		\end{pmatrix}-\left(\frac{\widetilde{K}}{4}\int f_1^0 \mathrm{d}{ v}+A \partial^2_x {S}^{,0}_1\right) J \begin{pmatrix}
			{S}^{}_2  \\
			{S}^{}_3
		\end{pmatrix} &=& 0,
\end{eqnarray}
where $J$ denotes the symplectic matrix 
$$
J=\begin{pmatrix}
	0 & 1 \\
-1 & 0
\end{pmatrix}.$$ 
By the eigen-decomposition 
 $$\begin{pmatrix}
	\frac{1}{2} & \frac{1}{2} \\
	\frac{1}{2}& -\frac{1}{2}
\end{pmatrix}
\begin{pmatrix}
	0 & 1 \\
	1 & 0
\end{pmatrix}
\begin{pmatrix}
	1 & 1 \\
	1 & -1 
\end{pmatrix}
=\begin{pmatrix}
1 & 0 \\
	0& -1
\end{pmatrix},$$ 
equation \eqref{HS1_a} can diagonalized to get  two transport equations that can be solved exactly in time 
\begin{equation}
\partial_t	\begin{pmatrix}
	\frac{1}{2}f_0+\frac{1}{2} f_1 \\
	\frac{1}{2}f_0-\frac{1}{2}f_1
\end{pmatrix}
-H \frac{\partial B_1^0}{\partial x}
\begin{pmatrix}
	1 & 0 \\
	0& -1
\end{pmatrix} \partial_v	\begin{pmatrix}
	\frac{1}{2}f_0+\frac{1}{2} f_1 \\
	\frac{1}{2}f_0-\frac{1}{2}f_1
\end{pmatrix} =0. 
\end{equation}
The exact solution for \eqref{HS1_b} is  
\begin{equation}
	\begin{pmatrix}
			f_2  \\
			f_3
		\end{pmatrix}(t,x,v)=\exp{\left(- B_1^0 J t\right)}\begin{pmatrix}
f_2^0(x,v)  \\
f_3^0(x,v)
\end{pmatrix},\  \text{with}\ \exp{(Js)}=\begin{pmatrix}
	\cos(s) & \sin(s) \\
-\sin(s) & \cos(s)
\end{pmatrix}.  
\end{equation}  
Similarly, we can get the exact solution for last system \eqref{HS1_c}
\begin{equation}\label{solutionofHS1}
	\begin{pmatrix}
			{S}^{}_2  \\
			{S}^{}_3
		\end{pmatrix}(t,x)=\exp{\left( \left( \frac{\widetilde{K}}{4}\int f_1^0 \mathrm{d}{ v}+A \partial^2_x {S}^{0}_1\right) J t\right)}\begin{pmatrix}
{S}^{0}_2 (x)  \\
			{S}^{0}_3 (x)
\end{pmatrix}.
\end{equation}

\subsection{Subsystem for $\mathcal{H}_{S_2}$}

The subsystem $\frac{\partial \mathcal{Z}}{\partial t} = \{ \mathcal{Z}, \mathcal{H}_{S_2} \}$ associated to $\mathcal{H}_{S_2}= H\int   f_2 B_2  \mathrm{d}x\mathrm{d}v+ AH \int (\frac{\partial {S_2}}{\partial x})^2 \mathrm{d}{ x}$ is

\begin{equation}\label{eq:HBS2}
\left\{
\begin{aligned}
&\frac{\partial f_0}{\partial t} = \{f_0, \mathcal{H}_{S_2} \} = H \frac{\partial B_2}{\partial x} \frac{\partial f_2}{\partial v}\\
&\frac{\partial f_1}{\partial t} = \{f_1, \mathcal{H}_{S_2} \} = B_2 f_3\\
&\frac{\partial f_2}{\partial t} = \{f_2, \mathcal{H}_{S_2} \} = H  \frac{\partial B_2}{\partial x} \frac{\partial f_0}{\partial v} \\
&\frac{\partial f_3}{\partial t} = \{f_3, \mathcal{H}_{S_2} \} = -B_2 f_1 \\
&\frac{\partial {S}^{}_1}{\partial t}  = \{ {S}^{}_1, \mathcal{H}_{S_2} \} = -\frac{\widetilde{K}}{4}S^{}_3 \int f_2 \mathrm{d}{ v} -A {S}^{}_3 \partial^2_x {S}^{}_2\\
&\frac{\partial {S}_2}{\partial t}  = \{ {S}^{}_2, \mathcal{H}_{S_2} \} = 0\\
&\frac{\partial {S}^{}_3}{\partial t}  = \{ {S}^{}_3, \mathcal{H}_{S_2} \} = \frac{\widetilde{K}}{4}S^{}_1 \int f_2 \mathrm{d}{ v} +A {S}^{}_1 \partial^2_x {S}^{}_2.
\end{aligned}
\right.
\end{equation}
with the initial value 
$(f_0^0(x,v),{\mathbf f}^0(x,v),{\mathbf S}^{0}(x))$ at time $t=0$ and ${\mathbf B}^0=-\frac{\widetilde{K} {\mathbf S^{0}}}{2}$. 
This subsystem is very similar to the $\mathcal{H}_{S_1}$ one, hence, 
as was done previously, we reformulate the equations by using ${S}^{}_2={S}^{0}_2$, $B_2=B_2^0$ and $\int f_2 \mathrm{d}{ v}=\int f_2^0 \mathrm{d}{ v}$
\begin{eqnarray}
\label{Hs2_a}
 \partial_t	
 \begin{pmatrix}
	f_0 \\
	f_2
\end{pmatrix}
-H \frac{\partial B_2^0}{\partial x} 
\begin{pmatrix}
	0 & 1 \\
	1 & 0
\end{pmatrix} 
\partial_v	
\begin{pmatrix}
f_0 \\
f_2
\end{pmatrix} &=&0, \\
\label{Hs2_b}
 \partial_t	
 \begin{pmatrix}
			f_1  \\
			f_3
		\end{pmatrix}
		- B_2^0 J 
		\begin{pmatrix}
f_1  \\
f_3
 \end{pmatrix} &=& 0,\\
 \label{Hs2_c}
  \partial_t
 	\begin{pmatrix}
			{S}^{}_1  \\
			{S}^{}_3
		\end{pmatrix}
		+\left(\frac{\widetilde{K}}{4}\int f_2^0 \mathrm{d}{ v}+A \partial^2_x {S}^{0}_2\right) J \begin{pmatrix}
			{S}^{}_1  \\
			{S}^{}_3
		\end{pmatrix} &=&0. 
\end{eqnarray}
As in the step ${\mathcal H}_{s_1}$, 
we have two transport equations from \eqref{Hs2_a} that can be solved exactly 
\begin{equation}
\partial_t	\begin{pmatrix}
	\frac{1}{2}f_0+\frac{1}{2} f_2 \\
	\frac{1}{2}f_0-\frac{1}{2}f_2
\end{pmatrix}
-H  \frac{\partial B_2^0}{\partial x}
\begin{pmatrix}
	1 & 0 \\
	0& -1
\end{pmatrix} \partial_v	\begin{pmatrix}
	\frac{1}{2}f_0+\frac{1}{2} f_2 \\
	\frac{1}{2}f_0-\frac{1}{2}f_2
\end{pmatrix} =0.
\end{equation}
Moreover, the exact solutions for the systems \eqref{Hs2_b} and \eqref{Hs2_c} are respectively
\begin{equation}
	\begin{pmatrix}
			f_1  \\
			f_3
		\end{pmatrix}(t,x,v)=\exp{\left( B_2^0 J t\right)}\begin{pmatrix}
f_1^0(x,v)  \\
f_3^0(x,v)
\end{pmatrix},\  \text{with}\ \exp{(Js)}=\begin{pmatrix}
	\cos(s) & \sin(s) \\
-\sin(s) & \cos(s)
\end{pmatrix}.  
\end{equation}  
and
\begin{equation}\label{solutionofHS2}
	\begin{pmatrix}
			{S}^{}_1  \\
			{S}^{}_3
		\end{pmatrix}(t,x)=\exp{\left(- \left( \frac{\widetilde{K}}{4}\int f_2^0 \mathrm{d}{ v}+A \partial^2_x {S}^{0}_2\right) J t\right)}\begin{pmatrix}
{S}^{0}_1 (x)  \\
			{S}^{0}_3 (x)
\end{pmatrix}.
\end{equation} 

\subsection{Subsystem for $\mathcal{H}_{S_3}$}

The subsystem $\frac{\partial \mathcal{Z}}{\partial t} = \{ \mathcal{Z}, \mathcal{H}_{S_3} \}$ associated to $\mathcal{H}_{S_3}= H \int   f_3 B_3  \mathrm{d}x\mathrm{d}v+ AH \int (\frac{\partial {S^{}_3}}{\partial x})^2 \mathrm{d}{ x}$ is

\begin{equation}\label{eq:HBS3}
\left\{
\begin{aligned}
&\frac{\partial f_0}{\partial t} = \{f_0, \mathcal{H}_{S_3} \} = H\frac{\partial B_3}{\partial x} \frac{\partial f_3}{\partial v}\\
&\frac{\partial f_1}{\partial t} = \{f_1, \mathcal{H}_{S_3} \} = -B_3 f_2 \\
&\frac{\partial f_2}{\partial t} = \{f_2, \mathcal{H}_{S_3} \} = B_3 f_1 \\
&\frac{\partial f_3}{\partial t} = \{f_3, \mathcal{H}_{S_3} \} = H \frac{\partial B_3}{\partial x} \frac{\partial f_0}{\partial v} \\
&\frac{\partial {S}^{}_1}{\partial t}  = \{ {S}^{}_1, \mathcal{H}_{S_3} \} = \frac{\widetilde{K}}{4} S^{}_2 \int f_3 \mathrm{d}{ v} +A {S}^{}_2 \partial^2_x {S}^{}_3\\
&\frac{\partial {S}^{}_2}{\partial t}  = \{ {S}^{}_2, \mathcal{H}_{S_3} \} = -\frac{\widetilde{K}}{4}S^{}_1 \int f_3 \mathrm{d}{ v} -A {S}^{}_1 \partial^2_x {S}^{}_3\\
&\frac{\partial {S}^{}_3}{\partial t}  = \{ {S}^{}_3, \mathcal{H}_{S_3} \} = 0.
\end{aligned}
\right.
\end{equation}
with the initial value 
$(f_0^0(x,v),{\mathbf f}^0(x,v),{\mathbf S}^{0}(x))$ at time $t=0$.
This subsystem is also very similar to the $\mathcal{H}_{{S_1}}$ one, hence, 
as was done previously, we reformulate the equations by using ${S}^{}_3={S}^{0}_3$, $B_3=B_3^0$ and $\int f_3 \mathrm{d}{ v}=\int f_3^0 \mathrm{d}{ v}$. The update of $(f_0, f_3)$ 
is performed by solving the following transport equation 
\begin{equation}
\partial_t	\begin{pmatrix}
	\frac{1}{2}f_0+\frac{1}{2} f_3 \\
	\frac{1}{2}f_0-\frac{1}{2}f_3
\end{pmatrix}
-H \frac{\partial B_3^0}{\partial x}
\begin{pmatrix}
	1 & 0 \\
	0& -1
\end{pmatrix} \partial_v	\begin{pmatrix}
	\frac{1}{2}f_0+\frac{1}{2} f_3 \\
	\frac{1}{2}f_0-\frac{1}{2}f_3
\end{pmatrix} =0.
\end{equation}
The exact solution for $(f_1, f_2)$ is 
\begin{equation}
	\begin{pmatrix}
			f_1  \\
			f_2
		\end{pmatrix}(t, x,v)=\exp{\left(-B_3^0 J t\right)}\begin{pmatrix}
f_1^0(x,v)  \\
f_2^0(x,v)
\end{pmatrix},\  \text{with}\ \exp{(Js)}=\begin{pmatrix}
	\cos(s) & \sin(s) \\
-\sin(s) & \cos(s)
\end{pmatrix}.  
\end{equation}  
and for $(S_1, S_2)$ we have 
\begin{equation}\label{solutionofHS3}
	\begin{pmatrix}
			{S}^{}_1  \\
			{S}^{}_2
		\end{pmatrix}(t,x)=\exp{\left( \left( \frac{\widetilde{K}}{4}\int f_3^0 \mathrm{d}{ v}+A \partial^2_x {S}^{0}_3\right) J t\right)}\begin{pmatrix}
{S}^{0}_1 (x)  \\
			{S}^{0}_2 (x)
\end{pmatrix}.
\end{equation}

\newpage

\bibliographystyle{abbrv}
\bibliography{biblio1}

\end{document}